\documentclass[10pt,reqno]{amsart}
\usepackage[utf8]{inputenc}
\usepackage[T1]{fontenc}
\usepackage{lmodern}
\usepackage{graphicx}
\usepackage{pgfplots}
\usepackage{tikz}
\usepackage{graphicx}
\usepackage{xcolor}
\usepackage{amsmath,amssymb,mathtools}
\usepackage{multicol}
\usepackage{hyperref}
\usepackage{cleveref}
\usepackage{dsfont}
\usepackage{lmodern}
\usepackage{mathrsfs}%
\usepackage{anyfontsize}
\usepackage{bbm}

\newtheorem{definition}{Definition}[section]
\newtheorem{theorem}{Theorem}[section]
\newtheorem{lemma}{Lemma}[section]

\newtheorem{proposition}{Proposition}[section]
\newtheorem{remark}{Remark}[section]
\allowdisplaybreaks
\numberwithin{equation}{section}
\newtheorem{assumption}{Assumption}

\newcommand{\beq}{\begin{equation}}
	\newcommand{\eeq}{\end{equation}}
\newcommand{\be}{\begin{eqnarray}}
	\newcommand{\ee}{\end{eqnarray}}
\newcommand{\noi}{\noindent}
\newcommand{\ba}{\begin{array}}
	\newcommand{\ea}{\end{array}}
\newcommand{\dis}{\displaystyle}

\newcommand{\sgn}{\textbf{sgn}}
\newcommand{\sat}{\operatorname{\mathbf{sat}}}
\newcommand{\eps}{\varepsilon}
\newcommand{\inte}{\operatorname{int}}

\title[Saturated stabilization of a fourth-order nonlinear parabolic equation]{Feedback stabilization of some fourth-order nonlinear parabolic equations with saturated controls}

\author[P. Guzm\'an]{Patricio Guzm\'an}
\address[P. Guzm\'an]{Departamento de Matem\'atica, Universidad Técnica Federico Santa Mar\'ia, Valpara\'iso, Chile.}
\email{patricio.guzmanm@usm.cl}
\author[F. Labra]{Felipe Labra}
\address[F. Labra]{Facultad de Matem\'aticas, Pontificia Universidad Cat\'olica de Chile, Santiago, Chile.}
\email{flabrai@estudiante.uc.cl}
\author[H. Parada]{Hugo Parada}
\address[H. Parada]{Universit\'e de Lorraine, CNRS, Inria, IECL, Nancy, France.}
\email{hugo.parada@inria.fr}

\thanks{P. Guzm\'an received partial financial support from FONDECYT 11240290 and  H. Parada is currently supported by the Agence Nationale de la Recherche through the QuBiCCS project (ANR-24-CE40-3008).}

\pgfplotsset{compat=1.18}

\begin{document}

\begin{abstract}
In this work, we analyze the internal and boundary stabilization of the Cahn-Hilliard and Kuramoto-Sivashinsky equations under saturated feedback control. We conduct our study through the spectral analysis of the associated linear operator. We identify a finite number of eigenvalues related to the unstable part of the system and then design a stabilization strategy based on modal decomposition, linear matrix inequalities (LMIs), and geometric conditions on the saturation function. Local exponential stabilization in $H^{2}$ is established.
\end{abstract}

\maketitle

\setlength\parindent{0pt} 
\section{Introduction}
\label{sec: intro}
We address the problem of feedback stabilization for two prototype fourth-order nonlinear parabolic equations: the Cahn–Hilliard (CH) and Kuramoto–Sivashinsky (KS) equations, both considered under saturated control actions. These models arise in different physical contexts but share similar mathematical structures and stabilization challenges. Regarding each model, we briefly mention: The CH equation was introduced in \cite{cahn,cahn-hilliard,cahn-hilliard-2} as a continuum model describing phase separation in an isotropic binary mixture when it is sufficiently cooled. Typical examples include binary alloys and binary solutions. Further details on the model can be found in \cite{elliot,novick,m2019} and the references therein. The KS equation was proposed in \cite{Kuramoto_1975,Kuramoto_1976} as a model for phase turbulence in reaction–diffusion systems, and later in \cite{ms,s} to describe the propagation of plane flames. It also captures the onset of instabilities in several physical and chemical systems; see, for instance, \cite{Chen_1986,Hooper_1996}.
\medskip

Let us now introduce the mathematical model.
Consider $L>0$ as the length of the spatial domain and $\lambda>0$ as the anti-diffusion parameter. Define the nonlinear term
\[\mathcal{N}(y)=\delta y\partial_{x}y-\nu\partial_{x}^{2}(y^{3}),\]
which allows us to recover CH term $\partial_{x}^{2}(y^{3})$ and the KS term $y\partial_{x}y$. The main system studied in this work is therefore given by

\beq\label{eq: KS}
\begin{cases}
\partial_{t}y+\lambda \partial_{x}^{2}y+\partial_{x}^{4}y+\mathcal{N}(y)=\mathcal{F}(y),& t>0, \ x \in (0,L), \\[1ex]
y(0,x)=y_{0}(x), & x \in (0,L),
\end{cases}\eeq
where $\mathcal{F}(y)$ denotes a suitable feedback control law with saturation.
Note that the system involves two nonlinear components $\mathcal{N}$ and $\mathcal{F}$, the latter arising from the saturation effect. The above system is complemented with one of the following three sets of boundary conditions
\begin{subequations}
\label{eq: bc KS}
\begin{align}
\label{eq: clamped}
&y(t,0)=y(t,L)=\partial_x y(t,0)=\partial_x y(t,L)=0,\\[1ex]
\label{eq: hinged}
&y(t,0)=y(t,L)=\partial_x^2 y(t,0)=\partial_x^2 y(t,L)=0,\\[1ex]
\label{eq: CH bc}
&\partial_x y(t,0)=\partial_x y(t,L)=\partial_x^3 y(t,0)=\partial_x^3 y(t,L)=0.
\end{align}
\end{subequations}

In the literature, the first set of boundary conditions \eqref{eq: clamped} is commonly referred to as \emph{clamped}, while \eqref{eq: hinged} is known as \emph{hinged}. Both are standard choices when analyzing KS equation. For CH equation, \eqref{eq: CH bc} is the boundary condition usually adopted, as it ensures mass conservation in the model describing phase separation of an isotropic binary mixture under negligible cooling effects and zero mass flux at the boundaries. An important part of the analysis will be devoted to the \emph{linear} version of \eqref{eq: KS}, which is

\beq\label{eq: LKS}
\begin{cases}
\partial_{t}y+\lambda \partial_{x}^{2}y+\partial_{x}^{4}y=\mathcal{F}(y),& t>0, \ x \in (0,L), \\[1ex]
y(0,x)=y_{0}(x), & x \in (0,L).
\end{cases}\eeq
The main idea relies on the spectral analysis of the associated spatial linear operator. This approach is commonly referred to in the literature as \emph{spectral reduction} or \emph{modal decomposition}. Typically, this framework is applied to Sturm–Liouville operators to stabilize the finite-dimensional unstable subspace. Such a strategy has been successfully employed in several contexts, for instance, for the semilinear one-dimensional heat equation \cite{coron2004global}, the multidimensional heat equation with disturbances \cite{guzman2020rapid}, CH equation \cite{bcgm2017,munteanu2019,cgm2021}, KS equation \cite{guzman2019stabilization,katz2021regional,lhachemi2023local}, and other related cases \cite{coron2006global,lhachemi2020feedback,katz2020constructive}. For the linear KS equation \eqref{eq: LKS} with clamped boundary conditions, \cite{Cerpa_2010} established null controllability in the case $L=1$ if and only if $\lambda \notin \mathcal{N}$, where
\beq\label{eq: critical}\mathcal{N}=\left\{\pi^{2}(k^{2}+\ell^{2}),\, \, k,\ell\in\mathbb{N}, \ \ 1\leq k<\ell, \ \ \text{ $k$ and $l$ have the same parity}\right\}.\eeq
In the same work, when $\lambda \notin \mathcal{N}$, an explicit feedback law was designed to ensure exponential stabilization in the $L^{2}(0,1)$ norm. Exponential stability for the hinged case was later analyzed in \cite{katz2021regional}. Additional controllability results were obtained in \cite{Cerpa_2016}, and null controllability on star-shaped networks was addressed in \cite{cazacu2018null}. From the stabilization viewpoint, several related studies have treated KS equation, including \cite{Christofides,Christofides_2000,liu2001stability,guzman2019stabilization,lhachemi2023local}. Control problems for CH type systems have been extensively investigated in the literature, mainly in the context of optimal and distributed control formulations \cite{Frigeri2016,Colli2018,Colli2023}. However, feedback stabilization of such models has received comparatively little attention.
Among the few contributions in this direction, the work \cite{bcgm2017} studied the local stabilization of a nonisothermal CH system around equilibrium states using a Riccati-based feedback law constructed from the unstable spectral modes of the associated linearized operator.
These results were further extended in \cite{marinoschi2018internal}. We also mention the recent works \cite{guzman2020local,azmi2023saturated}, which address local controllability and trajectory stabilization, respectively.
\medskip

In practical applications, actuator saturation is an unavoidable feature arising from physical or technological limitations. The analysis of such constraints in infinite-dimensional systems has attracted growing attention in recent years. For instance, \cite{Prieur_2016} investigated wave equations with both distributed and boundary saturated feedback laws, while more recent works have developed systematic Lyapunov-based approaches for saturated boundary stabilization of PDEs \cite{lhachemi2022local,lhachemi2022nonlinear}. Similar ideas have also been explored for dispersive models, such as the Korteweg–de Vries equation \cite{marx2017global,paradakdvsat,parada2022stability,parada2023stabilization}. To the best of our knowledge, however, no contribution has yet addressed the stabilization of fourth-order parabolic equations under input saturation. The closest related studies concern linear reaction–diffusion systems \cite{Mironchenko_2020} and the nonlinear Schlögl model \cite{azmi2023saturated}. For a recent and comprehensive review of saturation-based stabilization methods for PDEs, we refer to \cite{lhachemi2023saturated}. In this work, we consider the following saturation map. Let $\ell>0$ be the saturation level. We define the pointwise saturation map,
\begin{equation}
\label{eq: satl}
\sat(s)=\begin{cases}
s, &\text{ if } |s|\leq \ell,\\
\ell\sgn(s), &\text{ if }  |s|\geq \ell.
\end{cases}
\end{equation}
This saturation was considered in \cite{Mironchenko_2020,marx2017global}. The main contribution of this work is the analysis of CH and KS equations under saturated feedback, acting either internally or at the boundary, including stability results in $L^{2}$ and $H^{2}$ under mild controllability assumptions on the finite-dimensional unstable subspace. Furthermore, we establish the existence and uniqueness of solutions for the corresponding closed-loop systems. Most of the analysis is carried out in detail for the case of internal saturation, while the boundary case follows by adapting the same arguments. To the best of our knowledge, this is the first work addressing these equations under saturated control actions. The rest of the paper is organized as follows. In \Cref{sec: pre}, we introduce the notation and some preliminary results. \Cref{sec: caso memoria} is devoted to $L^{2}$ stabilization for the internal feedback case, while \Cref{sec: hg nl} establishes $H^{2}$ stabilization, which allows us to address the fully nonlinear problem. In \Cref{sec: bs}, we briefly explain how the same ideas can be extended to boundary feedback. Finally, some remarks and concluding perspectives are presented in \Cref{sec: conclusion}.

\section{Preliminaries}
\label{sec: pre}
We begin this section by introducing the linear operator associated with the spatial part of the system and recall some of its fundamental spectral properties:

\[\mathcal{A}: y \in D(\mathcal{A})\subset L^{2}(0,L) \longrightarrow -y^{\prime\prime\prime\prime}-\lambda y^{\prime\prime} \in L^{2}(0,L),\]
\[D(\mathcal{A})=\left\{y \in L^{2}(0,L), \, \mathcal{A}y\in L^{2}(0,L) \ \text{ and }  y \text{ satisfies } \eqref{eq: bc KS}\right\}.\]

In what follows, whenever no ambiguity arises, we omit the explicit mention of the boundary conditions. Hence, unless otherwise stated, all results apply to any of the admissible boundary conditions in \eqref{eq: bc KS}. By standard arguments in spectral theory, the operator $\mathcal{A}$ is self-adjoint and has a compact resolvent. Consequently, its spectrum $\sigma(\mathcal{A})$ consists only of real eigenvalues $\{\sigma_{k}\}_{k\in\mathbb{N}}$ satisfying $\sigma_{k}\to -\infty$ as $k\to\infty$. The associated eigenfunctions $\{e_{k}\}_{k\in\mathbb{N}}\subset D(\mathcal{A})$ form an orthonormal basis of $L^{2}(0,L)$.
\begin{remark}
For the \textit{hinged} boundary conditions  \eqref{eq: hinged}, the spectral analysis is straightforward.
Indeed, in this case $\mathcal{A}=-\mathcal{L}^{2}-\lambda\mathcal{L}$, where $\mathcal{L}$ denotes the Laplacian with Dirichlet boundary conditions. Thus, the eigenvalues and eigenfunctions are given by:
\[\sigma_{k}=\left(\frac{k\pi}{L}\right)^{2}\left(\lambda-\left(\frac{k\pi}{L}\right)^{2}\right), \qquad  e_{k}(x)=\sin{\left(\frac{k\pi x}{L}\right)}.\] 
A similar characterization holds when considering the Neumann Laplacian and boundary conditions \eqref{eq: CH bc}.
\end{remark}
Throughout this work, we identify matrices $M_{m\times n}$ with $\mathbb{R}^{m\times n}$ and are denoted by bold capital letters, $\mathbf{A}\in\mathbb{R}^{n\times n}$, $\mathbf{B}\in \mathbb{R}^{n\times m}$. The identity matrix is denoted by
$\mathbf{I}_n\in \mathbb{R}^{n\times n}$. We denote by $\mathbb{S}^n$ the space of $n\times n$ symmetric matrices. If $\mathbf{A}\in\mathbb{S}^n$, then $\mathbf{A}$ is orthonormal diagonalizable, and we can order its eigenvalues $\lambda_1 \geq \lambda_2 \geq ... \geq \lambda_n.$ In particular, we have
\[\det \mathbf{A} = \prod_{k=1}^n \lambda_k, \qquad \operatorname{tr} \mathbf{A} = \sum_{k=1}^n \lambda_k.\]
We consider $\mathbb{S}^n$ endowed with the \textit{Frobenius norm}:
\[\|\mathbf{A}\| = \left(\sum_{k=1}^n \lambda_k^2\right)^{1/2}.\]
For $z\in \mathbb{R}^n$, we consider the Euclidean norm $\left|\,\cdot\,\right|$. In that case, we have the following expressions for $\lambda_1$ and $\lambda_n$:
\[\lambda_1 = \sup_{z\not=0}\frac{z^\top \mathbf{A}z}{\left|z\right|^2}, \qquad \lambda_n = \inf_{z\not=0}\frac{z^\top \mathbf{A}z}{\left|z\right|^2}, \quad \text{ and } \quad \lambda_n \left|z\right|^2\leq z^\top \mathbf{A}z \leq \lambda_1 \left|z\right|^2, \qquad \forall z\in \mathbb{R}^n.\]
A matrix $\mathbf{A}\in \mathbb{S}^n$ is positive semi-definite if $z^\top \mathbf{A} z \geq 0$, $\forall z\in \mathbb{R}^n$. Given $\mathbf{A, B}\in \mathbb{S}^n$, we use the \textit{Loewner partial order} \cite[Section 1.10]{pukelsheim-1993} to write $\mathbf{A} \succeq \mathbf{B} \ \Longleftrightarrow \  \mathbf{A-B}\succeq 0$, which is equivalent to ask the eigenvalues of $\mathbf{A}-\mathbf{B}$ be non-negative. If $\mathbf{A}\in \mathbb{S}^n$ is positive semi-definite, we denote by $\mathbf{A}\succeq 0$. We say that $\mathbf{A}\in \mathbb{S}^n$ is positive definite, if
$z^\top \mathbf{A} z > 0, \ \forall z\in \mathbb{R}^n\setminus\{0\}$. Take $\mathbf{B}\in \mathbb{S}^n$ positive definite, then we have
$\mathbf{A} \succ \mathbf{B} \ \Longleftrightarrow \ \mathbf{A-B}\succ 0$.
We say that $\mathbf{A}\in \mathbb{R}^{n\times n}$ is Hurwitz if all its eigenvalues have negative real part.
\medskip

\noi  We consider the space $L^{2}(0,L)$ of square integrable functions on $(0,L)$, endowed with the inner product $\langle f,g\rangle_{L^{2}(0,L)}=\dis\int_{0}^{L}f(x)g(x)dx$ with associated norm, denoted by $\|\cdot\|_{L^{2}(0,L)}$. For an integer $s\geq1$, the $s$-order Sobolev space is denoted by $H^{s}(0,L)$ and is endowed with its usual norm, denoted by $\|\cdot\|_{H^{s}(0,L)}$. Let $X_{n} \subset L^{2}(0,L)$ be the subspace spanned by $(e_{j})_{j=1}^{n}$ and $\Pi_{n}$ be the orthogonal projection onto $X_{n}$, that is

\[\Pi_{n}y(t,\cdot)=\sum_{j=1}^{n}y_{j}(t,\cdot)e_{j}, \quad y_{j}(t,\cdot)=\langle y(t,\cdot),e_{j}\rangle_{L^{2}(0,L)}. \]

Let $X_{n}^{\perp}$ denotes the orthogonal complement of $X_{n}$ in $L^{2}(0,L)$ and $\iota: \mathbb{R}^{n} \rightarrow X_{n}$ be the isomorphism defined by $\iota(e^{j})=e_{j}$, where $(e^{j})_{j=1}^{n}$ is the canonical basis of $\mathbb{R}^{n}$. We recall that
\[\ell^{2}(\mathbb{N},\mathbb{R})=\left\{ (x_{k})_{k \in \mathbb{N}} \in \mathbb{R}^{\mathbb{N}} \ : \ \sum_{k=1}^{\infty}|x_{k}|^{2} < \infty \right\}. \]
We are going to use the isometric representation of $L^{2}(0,L)$ as $\ell_{2}(\mathbb{N},\mathbb{R})$ obtained by the isomorphism induced by $e_{j}(\cdot)\rightarrow e^{j}$, where $(e^{j})_{j \in \mathbb{N}}$ is the standard basis vectors in $\ell_{2}(\mathbb{N},\mathbb{R})$ and endowed with the standard $\ell_2$ norm. In this sense, we can write $L^{2}(0,L)=X_{n}\oplus X_{n}^{\perp}$, where $\oplus$ is the orthogonal sum of subspaces. We denote $\ell_{2}(\mathbb{N},\mathbb{R})=\mathbb{R}^{n}\oplus \ell_{2,j>n}$, where we identify $\mathbb{R}^{n}$ with the sequences with support in $\{1,\cdots,n\}$ and $\ell_{2,j>n}$ is the set of sequences in $\ell_{2}(\mathbb{N},\mathbb{R})$ which are $0$ in the first $n$ entries.
\medskip

To conclude this section, we state a Gronwall-type inequality that will be useful in the proof of stability for the fully nonlinear system. Let $b,k:[0,\infty)\to\mathbb{R}$ be two continuous functions. Let $p\ge0$ be a constant with $p\neq1$. Let $v(t)$ be a positive differentiable function such that

\begin{equation}
\label{ineq_1}
v'(t)\leq b(t)v(t)+k(t)v(t)^{p},~t\in[0,\infty).
\end{equation}
The following result can be found in \cite[Lemma 4.1]{bs1992}.
\begin{lemma}
\label{ineq} Let $q=1-p$ and consider,
\begin{equation}
\label{ineq_2}
w(t)=v(0)^{q}+q\int_{0}^{t}k(s)\mbox{exp}\left(-q\int_{0}^{s}b(\tau)\,d\tau\right)ds.
\end{equation}
\noindent If $w(t)>0$ for all $t\in[0,\infty)$, then
\begin{equation}
\label{ineq_3}
v(t)\leq\mbox{exp}\left(\int_{0}^{t}b(s)\,ds\right)w(t)^{1/q},~t\in[0,\infty).
\end{equation}
\end{lemma}
\section{Internal stabilization with saturated control}
\label{sec: caso memoria}
In this section, we design an internal feedback law, ensuring exponential stabilization of the linearized system under input saturation. The nonlinear effects will be treated later as perturbations of this linear dynamics in \Cref{sec: hg nl}. Inspired by the framework introduced in \cite{Mironchenko_2020}, we consider the controlled linear system
\beq\label{eq: main lineal}
\begin{cases}
\partial_{t}y+\lambda \partial_{x}^{2}y+\partial_{x}^{4}y=\dis\sum_{k=1}^{m}b_{k}(x)\sat(u_{k}(t)),& t>0, \ x \in (0,L), \\
y(0,x)=y_{0}(x), & x \in (0,L).
\end{cases}\eeq
Here $u: (0,\infty) \rightarrow \mathbb{R}^{m}$ denotes the control vector $u(t)=\begin{pmatrix}
  u_{1}(t),\hdots,u_{m}(t)
\end{pmatrix}^{\top}$,
and the vector function
$\mathbf{b}(x)=\begin{pmatrix}b_{1}(x),\dots,b_{m}(x)
\end{pmatrix}$, where $b_{k} \in L^{2}(0,L)$ describes the localization of the $k-$component of the control $u_{k}\in \mathbb{R}$. The saturation map acts component-wise, i.e.,
\[\sat(u(t))=\begin{pmatrix}
  \sat(u_{1}(t)),\hdots,\sat(u_{m}(t))
\end{pmatrix}^{\top}.\]
Formally, $y(t,\cdot)\in D(\mathcal{A})$ can be written as 
\[y(t,\cdot)=\sum_{j=1}^{\infty}y_{j}(t)e_{j}(\cdot),\]
where $y_{j}(t)$ are the associated Fourier components given by $\langle y(t,\cdot),e_{j}(\cdot)\rangle_{L^{2}(0,L)}$. Similarly, since $b_{k} \in L^{2}(0,L)$ for $k=1,\dots,m$
\[b_{k}(\cdot)=\sum_{j=1}^{\infty}b_{jk}e_{j}(\cdot), \quad b_{jk}=\langle b_{k}(\cdot),e_{j}(\cdot)\rangle_{L^{2}(0,L)}, \quad \forall j\in\mathbb{N}, \, \forall k=1,\dots,m.\]
Now, observe that \eqref{eq: main lineal} can be written as
\begin{equation}
\label{eq: y ope}
\partial_{t}y=\mathcal{A}y+\dis\sum_{k=1}^{m}b_{k}(x)\sat(u_{k}(t)).
\end{equation}
Using that $e_{j}$ are eigenfunctions of $\mathcal{A}$ we obtain the following infinite-dimensional system of ODEs
\begin{equation}
\label{eq: inf ode}
\dot{y}_{j}(t)=\sigma_{j}y_{j}(t)+\mathbf{b}_{j}\sat(u(t)), \qquad \forall j \in \mathbb{N},
\end{equation}
where $\mathbf{b}_{j}=\begin{pmatrix}
b_{j1}\ \cdots \ b_{jm}
\end{pmatrix} \in \mathbb{R}^{1\times m}$. By the asymptotic behaviour of the eigenvalues, we can choose $n \in \mathbb{N}$ as the number of non-negative eigenvalues and select $\eta>0$ such that
\beq\label{eq: neg eig}\sigma_{j}<-\eta<0, \qquad \forall j\geq n+1.\eeq
We then consider the projection of $y$ onto the unstable space. Let $z$, $\mathbf{A}$, and $\mathbf{B}$ be defined as follows

\beq\label{eq: matrices}z:=\begin{pmatrix}y_{1}\\\vdots\\y_{n}\end{pmatrix}, \qquad \mathbf{A}:= \begin{pmatrix}\sigma_{1}& \dots &0\\
\vdots&\ddots&\vdots\\
0&\dots&\sigma_{n}\end{pmatrix}, \qquad \mathbf{B}:=\begin{pmatrix}
b_{11}&\dots &b_{1m}\\
\vdots&\ddots&\vdots\\
b_{n1}&\dots&b_{nm}\end{pmatrix}. \eeq

Then, the dynamics of the first $n$ Fourier modes are governed by

\beq\label{eq: main z}
\dot{z}(t)=\mathbf{A}z(t)+\mathbf{B}\sat(u(t)).\eeq
We make the following assumption
\begin{assumption}
\label{assum: kalman}
The pair $(\mathbf{A},\mathbf{B})$ is stabilizable, i.e. there exists $\mathbf{K} \in \mathbb{R}^{m\times n}$, such that $\mathbf{A}+\mathbf{BK}$ is Hurwitz.
\end{assumption}
This assumption can be verified in many practical cases (see the discussion in \Cref{sec: par A-B}). In \cite[Example 1.1]{tarbouriech2011stability}, can be observed that even if $\mathbf{A}+\mathbf{BK}$ is Hurwitz, the presence of saturation in the feedback system $\dot{z}(t)=\mathbf{A}z(t)+\mathbf{B}\sat(\mathbf{K}z)$ may lead to the appearance of new equilibrium points and even diverging trajectories.
\begin{remark}
Note that even though the unstable part of the system \eqref{eq: main lineal} is finite-dimensional, the control acts on all the Fourier modes of the solution.
\end{remark}

We now aim to show that the exponential stability of \eqref{eq: inf ode} follows from the exponential stabilization of \eqref{eq: main z}. We look for the control $u$ in the form $u=\mathbf{K}z(t)$, and define the matrices $\mathbf{K}_{j}:=\mathbf{K}\iota^{-1}e_{j}$, $j=1,\dots,n$ and $\mathbf{K}:=(\mathbf{K}_{1},\dots,\mathbf{K}_{n})\in \mathbb{R}^{m\times n}$, with this choice, we obtain
\[\begin{aligned}
u(t)&=\mathbf{K}\iota^{-1}(\Pi_{n}y(\cdot,t))\\
&=\sum_{j=1}^{n}y_{j}\mathbf{K}\iota^{-1}e_{j}(\cdot)\\
&=\sum_{j=1}^{n}y_{j}\mathbf{K}_{j}(\cdot)\\
&=\mathbf{K}z(t).
\end{aligned}\]
With this feedback law, system \eqref{eq: inf ode} becomes the following infinite-dimensional ODE in $\ell_2(\mathbb{N},\mathbb{R})$
\begin{equation}
\label{eq: inf ode con}
\dot{y}_{j}(t)=\sigma_{j}y_{j}(t)+\mathbf{b}_{j}\sat(\mathbf{K}z(t)), \qquad \forall j \in \mathbb{N},
\end{equation}
and can be split as the finite-dimensional subsystem (the projection onto the unstable space)
\begin{equation}
\label{eq: z sat}
\dot{z}(t)=\mathbf{A}z(t)+\mathbf{B}\sat(\mathbf{K}z(t)),
\end{equation}
and the infinite-dimensional subsystem (the projection onto the stable space)
\begin{equation}
\label{eq: inf ode con stable}
\dot{y}_{j}(t)=\sigma_{j}y_{j}(t)+\mathbf{b}_{j}\sat(\mathbf{K}z(t)), \qquad  j\geq n+1.
\end{equation}
In this direction, \eqref{eq: inf ode con} is decomposed into a cascade system. The strategy is first stabilize the finite-dimensional part \eqref{eq: main z} and then use this information in the infinite-dimensional part \eqref{eq: inf ode con stable}.
\begin{remark}
System \eqref{eq: inf ode con} can also be written as
\[\dot{y}_{j}(t)=\sigma_{j}y_{j}(t)+\mathbf{b}_{j}\sat(\mathbf{K}\iota^{-1}(\Pi_{n}y)),  \qquad \forall j \in \mathbb{N}.\]
\end{remark}
Before proceeding with the proof of exponential stability, we recall some definitions concerning the region of attraction and local exponential stability.
\begin{definition}[Region of attraction] Assume that $\mathbf{K}$ is chosen such that $0$ is a locally asymptotically stable equilibrium of \eqref{eq: z sat}. We say that $S\subset \mathbb{R}^{n}$ is a region of attraction of $0$ if
\begin{enumerate}
    \item $0 \in \inte(S)$;
    \item for any $z_{0} \in S$ the corresponding solution of \eqref{eq: main z} satisfies $z(t,z_{0}) \rightarrow 0$ as $t \rightarrow \infty$;
    \item $S$ is forward invariant, i.e., for any $z_{0} \in S$ it holds that $z(t,z_{0}) \in S$ for all $t\geq 0$.
\end{enumerate}
By considering the partial order induced by set inclusion, the largest set satisfying the above properties is called the maximal region of attraction
\end{definition}

\begin{definition} System \eqref{eq: z sat} is called a locally exponentially in $0$ with region of attraction $S$, if the following conditions  hold:
\begin{enumerate}
    \item there exist $\varepsilon, M, \alpha > 0$ such that for all initial conditions with $|z_0|\le\varepsilon$, it holds
    \beq\label{eq: stb finsys}|z(t,z_{0})|\leq Me^{-\alpha t}|z_{0}|, \quad \forall t\geq 0.\eeq
    \item $\overline{B_{\eps}(0)}\subset S$ and $S$ is a region of attraction of \eqref{eq: z sat}.
\end{enumerate}
\end{definition}
The above definitions can be analogously extended to system \eqref{eq: inf ode con stable} and our PDE model \eqref{eq: main lineal}. With these tools at hand, we are now in position to state our main stability result of this section: consider the closed-loop system
\beq\label{eq: LKS closed loop}
\begin{cases}
\partial_{t}y+\lambda \partial_{x}^{2}y+\partial_{x}^{4}y=\dis\sum_{k=1}^{m}b_{k}(x)\sat\left(\left(\mathbf{K}\iota^{-1}(\Pi_{n}y(\cdot,t))\right)\right)_{k},& t>0, \ x \in (0,L), \\
y(0,x)=y_{0}(x), & x \in (0,L).
\end{cases}\eeq
\begin{theorem}
\label{th: sta lineal}
Let $\mathbf{K}$ be as in \Cref{assum: kalman}. For all $T>0$ and $y_{0}\in H^{2}(0,L)$, there exists a unique solution $y \in C([0,T];H^{2}(0,L))\cap L^{2}(0,T;H^{4}(0,L))$ of the closed-loop system \eqref{eq: LKS closed loop}. Moreover, this system is locally exponentially stable at $0$ with a region of attraction $\iota(\mathcal{R})\times X_{n}^{\perp}$ in the $L^{2}-$norm, where
\[\mathcal{R}:= \left\{z\in \mathbb{R}^{n} \ : \ z^{\top}\mathbf{P}z \leq 1 \right\}, \quad \text{ for some } \mathbf{P}\in\mathbb{R}^{n\times n}. \]
\end{theorem}
The proof of \Cref{th: sta lineal} is divided into two parts: in \Cref{sec: well posedness}, we address the local-in-time well-posedness using fixed-point arguments, \Cref{sec: global wp} is devoted to the global in time well-posedness and in \Cref{sec: stability}, we establish the exponential stability of the closed-loop system.
\subsection{Local well-posedness}
\label{sec: well posedness}
In this section, we study the well-posedness of the closed-loop system \eqref{eq: LKS closed loop}. We start by considering the following linear system with a source term:
\begin{equation}
\label{eq: LKS f}
\begin{cases}
\partial_{t}y+\partial_{x}^{4}y=f,& t\in(0,T), \ x \in (0,L), \\
y(0,x)=y_{0}(x), & x \in (0,L).
\end{cases}
\end{equation}
Let $B_T:= C([0,T];H^{2}(0,L))\cap L^{2}(0,T;H^{4}(0,L))$, endowed with the norm
\[\|f\|_{B_{T}}=\left(\|f\|_{C([0,T];H^{2}(0,L))}^{2}+\|f\|_{L^{2}(0,T;H^{4}(0,L))}^{2}\right)^{1/2}.\]
With some slight modifications of \cite[Proposition 1]{melendez2013lipschitz}, we have
\begin{proposition}
\label{prop: wp LKS f}
Let $T>0$, $y_{0}\in H^{2}(0,L)$ and $f\in L^{2}(0,T;L^{2}(0,L))$, then equation \eqref{eq: LKS f} has a unique (mild) solution $y \in B_{T}$ and there exists $C$ independent of $T$, such that
\beq\label{eq: est f}\|y\|_{B_{T}}\leq C\sqrt{1+T^{2}}\left(\|y_{0}\|_{H^{2}(0,L)}+\|f\|_{L^{2}((0,T);L^{2}(0,L))}\right).\eeq
Moreover, if $y_{0}\in D(\mathcal{A})$ and $f\in C^{1}([0,T];L^{2}(0,L))$, then  $y$ is a strong solution and $y \in C([0,T];D(\mathcal{A}))\cap C^{1}([0,T];L^{2}(0,L))$.
\end{proposition}
Now, for $\delta,\nu=0,1$, we consider the following system
\beq\label{eq: LKS o KS closed loop}
\begin{cases}
 \partial_{t}y+\lambda \partial_{x}^{2}y+\partial_{x}^{4}y+\delta y\partial_{x}y-\nu\partial_{x}^{2}(y^{3})=\dis\sum_{k=1}^{m}b_{k}(x)\sat\left(\left(\mathbf{K}\iota^{-1}(\Pi_{n}y(\cdot,t))\right)\right)_{k},& t>0, \ x \in (0,L), \\
y(0,x)=y_{0}(x), & x \in (0,L).
\end{cases}\eeq
System \eqref{eq: LKS o KS closed loop} reduces to \eqref{eq: LKS closed loop} when $\delta=\nu=0$ and coincides with the fully nonlinear closed-loop system \eqref{eq: main closed-loop nl} studied in \Cref{sec: nonlinear}. We first show local in time well-posedness of \eqref{eq: LKS o KS closed loop}.
\begin{proposition}
\label{prop: wp nlgKS}
Let $y_{0}\in H^{2}(0,L)$, then there exists $T^*>0$ such that equation \eqref{eq: LKS o KS closed loop} has a unique (mild) solution $y \in B_{T^*}$.
\end{proposition}
\begin{proof}
Throughout this part, $C$ will denote a positive constant, which may vary from line to line, but does not depend on $T$. Let $\theta\in(0,T]$ (to be chosen later). For the moment, assume that $y_{0}\in D(\mathcal{A})$ and $f\in C^{1}([0,T];L^{2}(0,L))$. Let $w \in B_{\theta}$ and consider the nonlinear operators
\[\begin{aligned}
&\mathcal{F}(w)=\dis\sum_{k=1}^{m}b_{k}(x)\sat\left(\left(\mathbf{K}\iota^{-1}(\Pi_{n}w(\cdot,t))\right)\right)_{k}, \quad \mathcal{G}(w)=-\lambda\partial_{x}^{2}w-\delta w\partial_{x}w+\nu\partial_{x}^{2}(y^{3}),\\
&\mathcal{H}(w)=\mathcal{F}(w)+\mathcal{G}(w).
\end{aligned}\]
Define the map $\Phi:B_{\theta} \rightarrow B_{\theta}$ by $\Phi(w)=y$, where $y$ is the unique solution of
\[\begin{cases}
\partial_{t}y=\mathcal{A}_{0}y+\mathcal{H}(w),& t\in(0,\theta), \\
y(0,\cdot)=y_{0}, &
\end{cases}\]
where $\mathcal{A}_{0}y=-y^{\prime\prime\prime\prime}$ and $D(\mathcal{A}_0)=D(\mathcal{A})$. Note that $y$ is well-defined. Indeed, if $w\in B_{\theta}$, using \eqref{eq: F est} and \eqref{eq: G est} we obtain $\mathcal{H}(w)\in L^{2}(0,\theta;L^{2}(0,L))$ and 
\beq\label{eq: H est}
\|\mathcal{H}(w)\|_{L^{2}(0,\theta;L^{2}(0,L))}\leq C\sqrt{\theta}\left(\|w\|_{B_{\theta}}+\delta\|w\|_{B_{\theta}}^{2}+\nu\|w\|_{B_{\theta}}^{3}\right)
\eeq
Using, \eqref{eq: est f} and \eqref{eq: H est}, we get
\[\begin{aligned}
\|\Phi(w)\|_{B_{T}}&\leq C\sqrt{1+\theta^{2}}\left(\|y_{0}\|_{H^{2}(0,L)}+\|\mathcal{H}(w)\|_{L^{2}((0,\theta);L^{2}(0,L))}\right) \\
&\leq C\sqrt{1+\theta^{2}}\|y_{0}\|_{H^{2}(0,L)}+C\sqrt{1+\theta^{2}}\sqrt{\theta}\left(\|w\|_{B_{\theta}}+\delta\|w\|_{B_{\theta}}^{2}+\nu\|w\|_{B_{\theta}}^{3}\right),
\end{aligned}\]
Now, we take $R=4C\|y_{0}\|_{H^{2}(0,L)}$ and consider $\Phi$ restricted to the ball $\mathcal{B}_{\theta,R}:=\{w \in B_{\theta}, \ \|w\|_{B_{\theta}}\leq R\}$. If $\theta$ is small enough such that
\[\begin{cases}
 \sqrt{1+\theta^{2}}<2,\\
C\sqrt{\theta}(1+\delta R+\nu R^{2})<\dfrac{1}{2},
\end{cases}\]
we get $\|\Phi(w)\|_{B_{T}}\leq R$, i.e., $\Phi$ maps $\mathcal{B}_{\theta,R}$ into itself. Now, for $w_{1}$, $w_{2} \in \mathcal{B}_{\theta,R}$, define $\Delta=\Phi(w_{1})-\Phi(w_{2})$, then $\Delta$ satisfies
\beq\label{eq: Delta}\begin{cases}
\partial_{t}\Delta=\mathcal{A}_{0}\Delta+\mathcal{H}(w_{1})-\mathcal{H}(w_{2}),& t\in(0,\theta), \\
\Delta(0,\cdot)=0. &
\end{cases}\eeq
Observe that
\[\begin{aligned}
\mathcal{H}(w_{1})-\mathcal{H}(w_{2})=&\mathcal{F}(w_{1})-\mathcal{F}(w_{2})-\lambda\partial_{x}^{2}\Delta-\delta\Delta\partial_{x}w_{2}-\delta w_{1}\partial_{x}\Delta+3\nu(w_{1})^{2}\partial_{x}^{2}\Delta\\
&+3\nu\Delta\partial_{x}^{2}w_{2}(w_{1}+w_{2})+6\nu\Delta(\partial_{x}w_{1})^{2}+6\nu\partial_{x}\Delta\partial_{x}(w_{1}+w_{2}),
\end{aligned}\]
therefore, using that the saturation map is Lipschitz continuous and similar to the previous estimates
\[\begin{aligned}
\|\mathcal{H}(w_{1})-\mathcal{H}(w_{2})\|_{L^{2}(0,\theta;L^{2}(0,L))}
\leq&C\sqrt{\theta}\left[\|\Delta\|_{B_{\theta}}+\delta\|\Delta\|_{B_{\theta}}\left(\|w_{1}\|_{B_{\theta}}+\|w_{2}\|_{B_{\theta}}\right)+\nu\|\Delta\|_{B_{\theta}}\left(\|w_{1}\|_{B_{\theta}}+\|w_{2}\|_{B_{\theta}}\right)^{2}\right].
\end{aligned}\]
Using this estimate and \eqref{eq: est f} for system \eqref{eq: Delta} we see
\[\begin{aligned}
\|\Phi(w_{1})-\Phi(w_{2})\|_{B_{\theta}}&\leq C\sqrt{1+\theta^{2}}\sqrt{\theta}\left\|\left[1+\delta\left(\|w_{1}\|_{B_{\theta}}+\|w_{2}\|_{B_{\theta}}\right)+\nu\left(\|w_{1}\|_{B_{\theta}}+\|w_{2}\|_{B_{\theta}}\right)^{2}\right]\Delta\right\|_{B_{\theta}}\\
&\leq  2C\sqrt{\theta}(1+2R+4R^{2})\|w_{1}-w_{2}\|_{B_{\theta}}.
\end{aligned}\]
By choosing $\theta$ small enough such that $2C\sqrt{\theta}(1+2R+4R^{2})<1$, we get that $\Phi$ is a contraction. Banach fixed-point theorem ensures the existence of unique solution $y\in B_{\theta}$ of \eqref{eq: LKS o KS closed loop}.   
\end{proof}
\begin{remark}
\label{rem: less regular}
A similar analysis can be done in a less regular framework, namely $y_{0}\in L^{2}(0,L)$ and $y \in C([0,T];L^{2}(0,L))\cap L^{2}(0,T;H^{2}(0,L))$. However, since our focus is on the stability of the fully nonlinear system, we need solutions in $y \in C([0,T];H^{2}(0,L))$.
\end{remark}
\subsection{A priori estimates}
\label{sec: global wp}
In this part, we use some a priori and stability estimates that will be used to extend the solutions of the closed-loop system for all positive time. Let $T>0$, we multiply \eqref{eq: LKS o KS closed loop} by $y$ and integrating over the spatial domain
\beq\label{eq: a priori nl L2}
\frac{1}{2}\frac{d}{dt}\left(\int_{0}^{L}|y|^{2}dx\right)+\int_{0}^{L}|\partial_{x}^{2}y|^{2}dx=-\lambda\int_{0}^{L}y\partial_{x}^{2}ydx+\int_{0}^{L}\mathcal{F}(y)ydx-\delta\int_{0}^{L}y^{2}\partial_{x}ydx+\nu\int_{0}^{L}\partial_{x}^{2}(y^{3})ydx.\eeq
Observe that, due to the boundary conditions, \eqref{eq: bc KS} the last term is non-positive,
\[\begin{aligned}
\int_{0}^{L}\partial_{x}^{2}(y^{3})ydx&=6\int_{0}^{L}(\partial_{x}y)^{2}y^{2}dx+3\int_{0}^{L}\partial_{x}^{2}yy^{3}dx\\
&=6\int_{0}^{L}(\partial_{x}y)^{2}y^{2}dx-9\int_{0}^{L}(\partial_{x}y)^{2}y^{2}dx\\
&\leq 0.
\end{aligned}\]
For the reminding nonlinear term, we perform the computation under boundary conditions \eqref{eq: CH bc}. In the other cases, the analysis is simpler because
\[\int_{0}^{L}y^{2}\partial_{x}ydx=\int_{0}^{L}\partial_{x}(y)^{3}dx=0.\]
First and second terms of \eqref{eq: a priori nl L2} are estimated as in \Cref{sec: well posedness}. For the third one, since $\partial_{x}y(t,\cdot) \in H_{0}^{1}(0,L)$, we obtain
\[\begin{aligned}
\delta\int_{0}^{L}y^{2}\partial_{x}ydx&\leq \|\partial_{x}y\|_{L^{\infty}(0,L)}\delta\int_{0}^{L}y^{2}dx\\
&\leq C \|\partial_{x}^{2}y\|_{L^{2}(0,L)}\delta\int_{0}^{L}y^{2}dx\\
&\leq \dfrac{1}{3}\|\partial_{x}^{2}y\|_{L^{2}(0,L)}^{2}+C\delta^{2}\left(\int_{0}^{L}y^{2}dx\right)^{2}.
\end{aligned}\]
Gathering all the estimates and multiplying by $2$, we observe
\beq\label{eq: dL2 nl}
\frac{d}{dt}\left(\int_{0}^{L}|y|^{2}dx\right)+\dfrac{2}{3}\int_{0}^{L}|\partial_{x}^{2}y|^{2}dx\leq C\int_{0}^{L}|y|^{2}dx+C\delta^{2}\left(\int_{0}^{L}y^{2}dx\right)^{2}.\eeq

Now multiply \eqref{eq: LKS o KS closed loop} by $\partial_{x}^{4}y$ to obtain
\beq\label{eq: a priori nl}
\frac{1}{2}\frac{d}{dt}\left(\int_{0}^{L}|\partial_{x}^{2}y|^{2}dx\right)+\int_{0}^{L}|\partial_{x}^{4}y|^{2}dx=-\lambda\int_{0}^{L}\partial_{x}^{2}y\partial_{x}^{4}ydx+\int_{0}^{L}\mathcal{F}(y)\partial_{x}^{4}ydx-\delta\int_{0}^{L}y\partial_{x}y\partial_{x}^{4}ydx+\nu\int_{0}^{L}\partial_{x}^{2}(y^{3})\partial_{x}^{4}ydx.\eeq

We now estimate right-hand side terms of the previous equality. For the first one
\[-\lambda\int_{0}^{L}y\partial_{x}^{2}\partial_{x}^{4}ydx\leq \dfrac{1}{5}\int_{0}^{L}|\partial_{x}^{4}y|^{2}dx+\dfrac{5\lambda^{2}}{2}\int_{0}^{L}|\partial_{x}^{2}y|^{2}dx.\]
For the second one
\[\begin{aligned}
 \int_{0}^{L}\mathcal{F}(y)\partial_{x}^{4}ydx&\leq \dfrac{1}{5}\int_{0}^{L}|\partial_{x}^{4}y|^{2}dx+\dfrac{5}{2}\int_{0}^{L}|\mathcal{F}(y)|^{2}dx\\
 &\leq \dfrac{1}{5}\int_{0}^{L}|\partial_{x}^{4}y|^{2}dx+C\int_{0}^{L}|y|^{2}dx\\
 &\leq \dfrac{1}{5}\int_{0}^{L}|\partial_{x}^{4}y|^{2}dx+C\|y\|^{2}_{H^{2}(0,L)}.
\end{aligned}\]
The third one, by the boundary conditions, either $y(t,\cdot)\in H_{0}^{1}(0,L)$ or $\partial_{x}y(t,\cdot)\in H_{0}^{1}(0,L)$
\[\begin{aligned}
 -\delta\int_{0}^{L}y\partial_{x}y\partial_{x}^{4}ydx&\leq \dfrac{1}{5}\int_{0}^{L}|\partial_{x}^{4}y|^{2}dx+\dfrac{5\delta^{2}}{2}\int_{0}^{L}|y\partial_{x}y|^{2}dx\\
 &\leq \dfrac{1}{5}\int_{0}^{L}|\partial_{x}^{4}y|^{2}dx+C\|y\|^{4}_{H^{2}(0,L)}.
\end{aligned}\]
For the last one,
\[\begin{aligned}
 \nu\int_{0}^{L}\partial_{x}^{2}(y^{3})\partial_{x}^{4}ydx&\leq \dfrac{1}{5}\int_{0}^{L}|\partial_{x}^{4}y|^{2}dx+\dfrac{5\nu^{2}}{2}\int_{0}^{L}|\partial_{x}^{2}(y^{3})|^{2}dx\\
 &\leq \dfrac{1}{5}\int_{0}^{L}|\partial_{x}^{4}y|^{2}dx+C\|y\|^{6}_{H^{2}(0,L)}.
\end{aligned}\]

Combining these estimates, and multiplying by $2$ we observe
\beq\label{eq: dH2 nl}
\frac{d}{dt}\left(\int_{0}^{L}|\partial_{x}^{2}y|^{2}dx\right)+\dfrac{2}{5}\int_{0}^{L}|\partial_{x}^{2}y|^{2}dx\leq C\|y\|^{2}_{H^{2}(0,L)}\left(1+\|y\|^{2}_{H^{2}(0,L)}+\|y\|^{4}_{H^{2}(0,L)}\right).\eeq
By adding \eqref{eq: dL2 nl} and \eqref{eq: dH2 nl}
\beq\label{eq: dH2 full nl}
\frac{d}{dt}\|y\|^{2}_{H^{2}(0,L)}+\int_{0}^{L}|\partial_{x}^{2}y|^{2}dx\leq C\|y\|^{2}_{H^{2}(0,L)}\left(1+\|y\|^{2}_{H^{2}(0,L)}+\|y\|^{4}_{H^{2}(0,L)}\right).\eeq
The differential inequality  \eqref{eq: dH2 full nl} combined with $H^2$ exponential stability will allow us to the solutions extended globally in time, yielding well-posedness for all $T>0$.

\subsection{Exponential stability}
\label{sec: stability}
We now address the stability part of \Cref{th: sta lineal}. We start with the following proposition that ensures that thanks to the cascade structure, it suffices to stabilize the finite-dimensional subsystem \eqref{eq: main z} to conclude the stability in the state space $\ell_{2}(\mathbb{N},\mathbb{R})$ for \eqref{eq: inf ode con} and $L^{2}(0,L)$ for \eqref{eq: LKS closed loop}.
\begin{proposition}
\label{prop: sta Lw}
Let $\mathbf{K}$ be as in \Cref{assum: kalman} and $S\subset \mathbb{R}^{n}$ the region of attraction of \eqref{eq: z sat} such that the system is locally exponentially stable in $0$. Then:
\begin{enumerate}
\item System \eqref{eq: inf ode con} (resp. \eqref{eq: LKS closed loop}) with the feedback law $u(t)=\mathbf{K}z(t)$ is locally exponentially stable in $0$ with region of attraction $S\times \ell_{2,j>N}$ (resp. $\iota(S)\times X_{N}^{\perp}$).

\item In addition, for any closed and bounded set $G \subset \inte(\iota(S)\times X_{N}^{\perp})$, there exist two positive constants $C$, $\alpha$ such that for any initial condition $y_{0}\in G$, the solution $y$ of \eqref{eq: LKS closed loop} satisfies
\beq\label{eq: decay L2 main}\|y(t,\cdot)\|_{L^{2}(0,L)} \leq Me^{-\alpha t}\|y_{0}\|_{L^{2}(0,L)}.\eeq
\end{enumerate}
\end{proposition}
\begin{proof}
Let $\mathbf{K}$ as in \Cref{assum: kalman} and $S\subset \mathbb{R}^{n}$ the region of attraction of \eqref{eq: main z} such that the system is locally exponentially stable in $0$. Take $G^{\prime} \subset \inte(S)$ bounded and closed. Then, there exist $M,a>0$ such that for all $z_{0} \in G^{\prime}$ the solution of \eqref{eq: main z} satisfies

\beq\label{eq: exp ode sat}
|z(t)|\leq Me^{-at}|z_{0}|, \qquad t\geq 0.
\eeq
Applying the feedback law $u(t)=\mathbf{K}z(t)$ and Duhamel's formulae for \eqref{eq: inf ode con stable}
\[y_{j}(t)=e^{\sigma_{j}t}y_{j}(0)+\mathbf{b}_{j}\int_{0}^{t}e^{\sigma_{j}(t-s)}\sat(\mathbf{K}z(s))ds, \qquad \forall j\geq n+1.\]
Observe that, using the definition of the saturation map \eqref{eq: satl} we get $|\sat(\mathbf{K}z)|\leq |\mathbf{K}z|\leq \|\mathbf{K}\||z|$. Therefore, for all $j\geq n+1$, using \eqref{eq: neg eig} and \eqref{eq: exp ode sat} we get
\[\begin{aligned}
|y_{j}(t)|&\leq e^{-\eta t}|y_{j}(0)|+M|\mathbf{b}_{j}|\int_{0}^{t}e^{-\eta(t-s)}e^{-as}\|\mathbf{K}\||z_{0}|ds,\\
&=e^{-\eta t}|y_{j}(0)|+M|\mathbf{b}_{j}|\|\mathbf{K}\||z_{0}|\left(\dfrac{e^{-at}-e^{-\eta t}}{\eta-a}\right),
\end{aligned}\]
where the last term is positive due to the monotonicity of the exponential function. Therefore,
\[
|y_{j}(t)|\leq 2e^{-2\eta t}|y_{j}(0)|^{2}+2(M\|\mathbf{K}\|)^{2}|\mathbf{b}_{j}|^{2}|z_{0}|^{2}\left(\dfrac{e^{-at}-e^{-\eta t}}{\eta-a}\right)^{2}.\]
Without loos of generality, assume $\eta>a$,
\[\left(\dfrac{e^{-at}-e^{-\eta t}}{\eta-a}\right)^{2}\leq \dfrac{e^{-2at}}{(\eta-a)^{2}}.\]
Summing over, $j\geq n+1$, we obtain the following estimate
\[\sum_{j\geq n+1}|y_{j}(t)|\leq 2e^{-2at}\sum_{j\geq n+1}|y_{j}(0)|^{2}+2\left(\dfrac{M\|\mathbf{K}\|}{(\eta-a)}\right)^{2}\left(\sum_{j\geq n+1}|\mathbf{b}_{j}|^{2}\right)e^{-2at}|z_{0}|^{2}.\]
By definition and Bessel inequality
\[\sum_{j\geq n+1}|\mathbf{b}_{j}|^{2}=\sum_{j\geq n+1}\sum_{k=1}^{m}|b_{jk}|^{2}\leq \|\mathbf{b}\|^{2}_{(L^{2}(0,L))^{m}}.\]
Finally, recalling \eqref{eq: exp ode sat} 
\[|z(t)|^{2}=\sum_{j=1}^{n}|y_{j}(t)|^{2}, \quad \text{ and } \quad |z_{0}|^{2}=\sum_{j=1}^{n}|y_{j}(0)|^{2},\]
we deduce
\[\sum_{j\in \mathbb{N}}|y_{j}(t)|\leq Ce^{-2at}\sum_{\in \mathbb{N}}|y_{j}(0)|, \qquad C=\max\left\{2,M^{2},2\left(\dfrac{M\|\mathbf{K}\|\|\mathbf{b}\|^{2}_{(L^{2}(0,L))^{m}}}{(\eta-a)}\right)^{2}\right\}.\]
This corresponds to exponential decay in the state space $\ell_{2}(\mathbb{N},\mathbb{R})$ for initial data in $G^{\prime}\times\ell_{2, j>n}$. To conclude the proof, it is enough to take $G^{\prime} \subset \inte(S)$ such that $0 \in \inte(G^{\prime})$ and thus $0 \in \inte(G^{\prime})\times\ell_{2, j>n}$. For the second part, using the isomorphism $\iota:\mathbb{R}^{n}\mapsto X_{n}$, the isometric isomorphism between $\ell_{2, j>n}$ and $X_{n}^{\perp}$ and Parseval inequality we obtain the local exponential decay in $0 \in L^{2}(0,L)$  with region of attraction $\iota(G^{\prime})\times X_{n}^{\perp} \subset L^{2}(0,L)$.
\end{proof}

\begin{remark}
The previous proof is analogues if we have now $\eta<a$ considering $e^{-at}<e^{-\eta t}$. In that case, we get
\[\sum_{j\in \mathbb{N}}|y_{j}(t)|\leq Ce^{-2\eta t}\sum_{j\in \mathbb{N}}|y_{j}(0)|.\]
If $\eta=a$, since the spectrum of $\mathcal{A}$ is discrete, we can always consider $\Tilde{\eta}$ close to $\eta$ such that $\Tilde{\eta}\neq a$, and \eqref{eq: neg eig} still holds with $\Tilde{\eta}$.
\end{remark}

The main hypothesis in the proof of the previous proposition is that \eqref{eq: z sat} is locally exponentially stable. As we already mentioned, even if $\mathbf{A}+\mathbf{B}\mathbf{K}$ is Hurwitz, \eqref{eq: z sat} could not be globally stable. To show that \eqref{eq: z sat} is locally exponentially stable, we use the theory of LMIs and appropriate geometric properties of the saturation. This strategy was used in \cite{Mironchenko_2020}, since our system is different from the one analyzed in \cite{Mironchenko_2020}, the matrices $\mathbf{A}$, $\mathbf{B}$ and $\mathbf{K}$ are not the same too, but they have the same structure. The spirit of the proof is the following. It is well known that to ensure the exponential stability of a linear system $\dot{z}=\mathbf{J}z$ it is enough to find a positively defined matrix $\mathbf{P}$ satisfying the Lyapunov inequality,
\[\mathbf{J}^{\top}\mathbf{P}+\mathbf{P}\mathbf{J}\prec -\mathbf{I}_{n},\]
which corresponds to a system of $\frac{n(n-1)}{2}$ equations that can be solved analytically and numerically. The action of the saturation map makes the task more intricate yet still allows us to guarantee local stability.
\medskip

For the saturation map $\sat$, we introduce the \textit{deadzone} nonlinearity $\phi:\mathbb{R}^{m}\mapsto \mathbb{R}^{m}$ by
\[\phi(u)=\sat(u)-u, \qquad \forall u \in \mathbb{R}^{m},\]
we also use the following generalized sector condition \cite[Lemma 1.6]{tarbouriech2011stability}.
\begin{lemma}
\label{lem: gen sector}
If for some $\mathbf{C} \in \mathbb{R}^{m\times n}$, $z\in \mathbb{R}^{n}$ and $j\in\{1,\dots,m\}$ it holds that $|((\mathbf{K}-\mathbf{C})z)_{j}|\leq \ell$, then
\[\phi_{j}(\mathbf{K}z)(\phi_{j}(\mathbf{K}z)+(\mathbf{C}z)_{j})\leq 0.\]
\end{lemma}
As was noticed in \cite[Appendix C.5]{tarbouriech2011stability}, the above inequality remains true when we multiply by positive scalars. In particular, for $\mathbf{D}\in \mathbb{R}^{m\times m}$ a diagonal, positive defined matrix, we  have
\[\phi_{j}(\mathbf{K}z)\mathbf{D}(\phi_{j}(\mathbf{K}z)+(\mathbf{C})z)_{j})\leq 0.\]
The motivation behind using the generalized sector condition is that allows to construct a Lyapunov functional that will help us to study exponential decay thanks to the appearance of an augmented matrix involving LMIs.

\noi Another useful tool is the Schur complement \cite[Appendix C.5]{tarbouriech2011stability}.

\begin{lemma}(Schur's lemma) \label{Lema de Schur}
Let $\mathbf{A}\in \mathbb{R}^{n\times n}$, $\mathbf{B}\in \mathbb{R}^{m\times n}$, $\mathbf{C}\in \mathbb{R}^{m\times m}$ and consider
$$\mathbf{M}:=\begin{pmatrix}
    \mathbf{A} & \mathbf{B}^\top\\
    \mathbf{B} & \mathbf{C}
\end{pmatrix}.$$
If $\mathbf{C}$ is positive definite, then $\mathbf{M}$ is positive definite if and only if its Schur complement $\mathbf{M\setminus C:= A-B^\top C^{-1}B}$ is positive semi-definite.
\end{lemma}

The next proposition shows the local exponential decay of \eqref{eq: z sat} and gives us an explicit region of attraction \cite[Proposition 2]{Mironchenko_2020}.
\begin{proposition}
\label{prop: region fed}
\sloppy Let $\mathbf{K}$ be as in \Cref{assum: kalman}. Let $\mathbf{P}\in\mathbb{R}^{n\times n}$ be symmetric and positive definite, and $\mathbf{D}\in\mathbb{R}^{n\times n}$ be diagonal and positive definite and $\mathbf{C}\in\mathbb{R}^{m\times n}$ such that

\beq\label{eq: M1} M_{1}:= \begin{pmatrix}
(\mathbf{A}+\mathbf{BK})^{\top}\mathbf{P}+\mathbf{P}(\mathbf{A}+\mathbf{BK}) & \mathbf{P}\mathbf{B}-(\mathbf{D}\mathbf{C})^{\top}\\
(\mathbf{P}\mathbf{B})^{\top}-\mathbf{D}\mathbf{C} & -2\mathbf{D}
\end{pmatrix} \prec 0,\eeq
and
\beq\label{eq: M2} M_{2}:= \begin{bmatrix}
\mathbf{P} & (\mathbf{K}-\mathbf{C})^{\top}\\
\mathbf{K}-\mathbf{C} & \ell^{2}\mathbf{I}_{m}
\end{bmatrix} \succeq 0.\eeq
Then \eqref{eq: z sat} is locally exponentially stable in $0$ with region of attraction $\mathcal{R}$ given by
\beq\label{eq: R} \mathcal{R}:= \left\{z\in \mathbb{R}^n \ : \ z^{\top}\mathbf{P}z \leq 1 \right\}. \eeq

Moreover, in $\mathcal{R}$, the function $V_{1}(z):=z^{\top}\mathbf{P}z$, $z \in \mathbb{R}^{n}$, decreases exponentially fast to $0$ along the solutions to \eqref{eq: z sat}, i.e., there is a constant $\alpha>0$ so that

\beq\label{eq: V1} \dot{V_{1}}(z)\leq -\alpha|z|^{2}, \quad z \in \mathcal{R}.\eeq
\end{proposition}
\begin{proof}
Let us start by proving that the LMIs \eqref{eq: M1} and \eqref{eq: M2} are feasible, i.e., that if we take $\mathbf{K}$ as in \Cref{assum: kalman}, then it is possible to find $\mathbf{P}\in \mathbb{R}^{n\times n}$ symmetric positive definite, $\mathbf{D}\in \mathbb{R}^{m\times m}$ diagonal, positive definite and  $\mathbf{C}\in \mathbb{R}^{m\times n}$ such that \eqref{eq: M1} and \eqref{eq: M2} hold. As $\mathbf{A+BK}$ is Hurwitz, by the Lyapunov inequality \cite[equation (1.2)]{boyd-1994}, there exists a symmetric positive definite matrix $\mathbf{P}\in \mathbb{R}^{n\times n}$ such that
$$\left(\mathbf{A+BK}\right)^\top \mathbf{P} + \mathbf{P}\left(\mathbf{A+BK}\right) \prec 0.$$

To ensure the existence of such a $\mathbf{C}$, take $\mathbf{C}=\mathbf{0}\in \mathbb{R}^{m\times n}$ and observe that $\ell^2\mathbf{I}_m \succ 0$. Consider $\mathbf{P}\in \mathbb{R}^{n\times n}$ with $\|\mathbf{P}\|$ sufficiently large such that $$\mathbf{P} - \mathbf{K}^\top \frac{1}{\ell^2}\mathbf{I}_m \mathbf{K} \succeq 0.$$

\noi and $\mathbf{D}\in \mathbb{R}^{m\times m}$ with $\|\mathbf{D}\|$ large enough. It follows from \eqref{Lema de Schur} that \eqref{eq: M1} and \eqref{eq: M2} are feasible.  As \eqref{eq: M1} and  \eqref{eq: M2}  are feasible, we can consider $\mathbf{P}, \mathbf{D}, \mathbf{C}$, not necessarily the same as before satisfying \eqref{eq: M1} and  \eqref{eq: M2}.  Then,
\begin{align*}
\dot{V}_1(z) &= \dot{z}^\top \mathbf{P}z+z^\top \mathbf{P}\dot{z}\\
&= (\mathbf{A}z+\mathbf{B}\sat{\mathbf{K}z)}^\top \mathbf{P}z + z^\top \mathbf{P} (\mathbf{A}z+\mathbf{B}\sat{\mathbf{K}z)},
\end{align*}
introducing the \emph{deadzone} nonlinearity, we get
\[\sat{\mathbf{K}z} = \phi(\mathbf{K}z)+\mathbf{K}z,\]
thus
\[\dot{V}_1(z) = z^\top \left[(\mathbf{A+BK})^\top \mathbf{P} + \mathbf{P}(\mathbf{A+BK})\right]z + 2z^\top \mathbf{\mathbf{P}B}\phi(\mathbf{K}z),\]

\noi where we have used $z^\top \mathbf{PB}\phi(\mathbf{K}z) = \phi(\mathbf{K}z)^\top \mathbf{(PB)^\top}z$. Suppose that $|((\mathbf{K-C})z)_j| \leq \ell$ for all  $j\in \{1,...,m\}$ by \Cref{lem: gen sector} we obtain
\[0 \leq -2\phi(\mathbf{K}z)^\top \mathbf{D}(\phi(\mathbf{K}z)+\mathbf{C}z)),\]
therefore
\[\begin{aligned}
    \dot{V}_1(z) &=z^\top \left[(\mathbf{A+BK})^\top \mathbf{P} + \mathbf{P}(\mathbf{A+BK})\right]z + 2z^\top \mathbf{\mathbf{P}B}\phi(\mathbf{K}z)\\
    &\leq z^\top \left[(\mathbf{A+BK})^\top \mathbf{P} + \mathbf{P}(\mathbf{A+BK})\right]z + 2z^\top \mathbf{PB}\phi(\mathbf{K}z)-2\phi(\mathbf{K}z)^\top \mathbf{D}(\phi(\mathbf{K}z)+\mathbf{C}z))\\
    &= \begin{pmatrix}
        z\\ \phi(\mathbf{K}z)
    \end{pmatrix}^\top \mathbf{M}_1 \begin{pmatrix}
        z\\ \phi(\mathbf{K}z)
    \end{pmatrix}.
\end{aligned}\]

\noi As $\mathbf{M}_1$ is negative definite, we can take $-\alpha<0$ as the largest eigenvalue of $\mathbf{M}_1$, it follows that $\mathbf{M}_1 \preceq -\alpha \mathbf{I}_{n+m}$, from where we get

\begin{equation}\label{decaimientoexponencialv1}
    \dot{V}_1(z) \leq -\alpha |z|^2.
\end{equation}

Let $\beta_{\text{max}}>0$ be the largest eigenvalue of $\mathbf{P}\in \mathbb{R}^{n\times n}$,
\[-\alpha \left|z\right|^2\leq -\frac{\alpha}{\beta_{\text{max}}}z^\top \mathbf{P}z = -\frac{\alpha}{\beta_{\text{max}}}V_1(z).\]

We conclude
$$\dot{V}_1(z) \leq -\frac{\alpha}{\beta_{max}}V_1(z).$$

Finally, we must prove that $|((\mathbf{K-C})z)_j|\leq \ell$ for all $j\in \{1,...,m\}$. In that direction, we force to the attraction region to be included in the generalized sector

\[\mathcal{R}=\{z\in \mathbb{R}^n \ : \ z^\top \mathbf{P}z\leq 1\} \subseteq \{z\in \mathbb{R}^n \ : \ |(\mathbf{K-C})z|\leq \ell\}.\]

This can be rewritten as

\beq\label{eq: set inclusion}\{z\in \mathbb{R}^n \ : \ z^\top \mathbf{P}z\leq 1\} \subseteq \{z\in \mathbb{R}^n \ : \ z^\top(\mathbf{K-C})^\top \frac{1}{\ell^2}\mathbf{I}_m(\mathbf{K-C})z\leq 1\}.\eeq

\noi As both sets are ellipsoids centered at $0\in \mathbb{R}^n$,  \cite[Proposition 1]{elipsoide} \eqref{eq: set inclusion} is equivalent to the following matrix inequality
\beq\label{eq: set ine}\mathbf{M}_2\setminus \ell^2 \mathbf{I}_m= \mathbf{P} - (\mathbf{K-C})^\top \frac{1}{\ell^2}\mathbf{I}_m(\mathbf{K-C}) \succeq 0,\eeq

where $\mathbf{M}_2\setminus \ell^2 \mathbf{I}_m$ denotes the Schur's complement. As ${\ell^2}\mathbf{I}_m$ is positive definite, by \Cref{Lema de Schur}, \eqref{eq: set ine} is equivalent to ask $\mathbf{M}_2$ satisfies \eqref{eq: M2}, i.e.,
\begin{equation*}
    \mathbf{M}_2=\begin{pmatrix}
    \mathbf{P} & (\mathbf{K}-\mathbf{C})^\top\\
    \mathbf{K}-\mathbf{C} & \ell^2 \mathbf{I}_m
\end{pmatrix}\succeq 0.
\end{equation*}
\end{proof}

We can now easily deduce our stabilization result for the system \eqref{eq: main lineal}. Let $\mathbf{K}$ as in \Cref{assum: kalman}, then taking  $\mathbf{P}$,  $\mathbf{D}$ and $\mathbf{C}$ as in \Cref{prop: region fed}, then \eqref{eq: z sat} is locally exponentially stable in $\mathcal{R}:= \left\{z\in \mathbb{R}^n \ : \ z^{\top}\mathbf{P}z \leq 1 \right\}$. By \Cref{prop: sta Lw} \eqref{eq: LKS closed loop} is locally exponentially stable in $0\in L^{2}(0,L)$ with region of attraction $\iota(\mathcal{R})\times X_{n}^{\perp}$. This completes the proof of \Cref{th: sta lineal}. $\square$
\subsection{On the controllability \texorpdfstring{\Cref{assum: kalman}}{Assumption~\ref{assum: kalman}}}
\label{sec: par A-B}
In the context of control involving Sturm-Liouville operators, it is standard to work under \Cref{assum: kalman}, i.e., the pair $(\mathbf{A, B})$ is stabilizable (controllable). In this part, we wish to go a bit further in this assumption in the specific case of the spatial operator associated to the linear KS equation with clamped boundary conditions \eqref{eq: clamped}.

\begin{enumerate}
    \item Case $m=1$. With a single spatial localization $b$, we have
    \[\mathbf{B}=\begin{pmatrix}
        b_1\\
        \vdots \\
        b_n
    \end{pmatrix}.\]
    \noi Therefore, the Kalman condition for the pair $(\mathbf{A, B})$ is seen as
    \begin{align*}
        \operatorname{det}\left(\mathbf{B}, \mathbf{AB}, \cdots, \mathbf{A}^{n-1}\mathbf{B}\right) &= \operatorname{det}\begin{pmatrix}
            b_1 & \sigma_1 \, b_1 & \cdots & \sigma_1^{n-1}\,b_1\\
            \vdots & \vdots & \ddots & \vdots\\
            b_n & \sigma_n \, b_n & \cdots & \sigma_n^{n-1}\,b_n
        \end{pmatrix}\\
        &=\displaystyle{\prod_{j=1}^{n} b_j} \operatorname{det}\begin{pmatrix}
            1 & \sigma_1  & \cdots & \sigma_1^{n-1}\\
            \vdots & \vdots & \ddots & \vdots\\
            1 & \sigma_n  & \cdots & \sigma_n^{n-1}
        \end{pmatrix}\\
        &= \displaystyle{\prod_{j=1}^{n} b_j} \text{Vdm}_{\sigma_{1},\cdots,\sigma_{n}}.
    \end{align*}
    \noi where $\text{Vdm}_{\sigma_{1},\dots,\sigma_{n}}=\prod_{1 \leq i < k \leq n} (\sigma_k-\sigma_i)$ corresponds to the Vandermonde determinant. Hence, controllability of the pair $(\mathbf{A, B})$ holds if $b_{j}\neq 0$ for all $j\in \{1,..., n\}$ and when the eigenvalues are simple. This last condition holds provided that $\lambda\notin\mathcal{N}$.

    \item Case $m\geq 2$ and simple eigenvalues. Using the Popov-Belevitch-Hautus test to study the block associated to each eigenvalue. In this case the matrix $\mathbf{B}$ has the form
    \[\mathbf{B}=\begin{pmatrix}
        b_{11} & \cdots & b_{1m}\\
        \vdots & \ddots & \vdots\\
        b_{n1} & \cdots & b_{nm}
    \end{pmatrix},\]
    For each eigenvalue, we aim to obtain
    $$\operatorname{Rank}\begin{bmatrix}
        \mathbf{A}-\sigma_j \mathbf{I}_n & \vrule &  \mathbf{B}
    \end{bmatrix} = n, \qquad \forall j\in \{1,...,n\}.$$
    Thus, the pair $(\mathbf{A, B})$ is controllable if 
    $$\forall j\in \{1, \ldots , n\}, \ \exists k\in \{1,\ldots , m\}:\ b_{jk} \not=0.$$

    \item Case $m\geq 1$ with repeated eigenvalues. When $m=1$, the controllability may fail when the Kalman matrix
    \[\mathbf{K}_{\mathbf{A,B}} = [\mathbf{B} \,|\, \mathbf{AB} \,|\, \cdots \,|\, \mathbf{A}^{n-1}\mathbf{B}] \in \mathbb{R}^{n\times nm},\]
    satisfies rank$(\mathbf{K})<n$. If we consider the localizations of the controller as design functions, a possible approach is to introduce additional control localizations, i.e., to consider (in the unsaturated case)
    \[ \partial_{t}y = \mathcal{A}y + \sum_{k=1}^{m+N_0}b_ku_k.\]
    In that framework, the Kalman matrix takes the form
    \[\tilde{\mathbf{K}}_{\mathbf{A,B}} = [\mathbf{B} \,|\, \mathbf{AB} \,|\, \cdots \,|\, \mathbf{A}^{n-1}\mathbf{B}] \in \mathbb{R}^{n\times n(m+N_{0})},\]
    where $N_0\in \mathbb{N}$ is the number of additional localizations. Since the column operations do not change the rank of the matrix, we can assume that
    \[\tilde{\mathbf{K}}_{\mathbf{A,B}} = [\mathbf{K}_{\mathbf{A,B}} \ |\ \mathbf{N}].\]
     We need an appropriated number of additional localization to complement the rank of $\mathbf{K}_{\mathbf{A,B}}$. As shown in \cite{lhachemi-2023vp}, one possible choice is
    \[N_0 = \max\{d_1, \ldots , d_n\},\]
   where $d_j$ is the algebraic multiplicity of $\sigma_j$.

    \begin{remark}
        The previous condition is clearly not sharp. For instance take,

        \[\mathbf{A} = \begin{bmatrix}
            2 & 0 & 0\\
            0 & 2 & 0\\
            0 & 0 & 1
        \end{bmatrix}, \qquad \mathbf{B} = \begin{bmatrix}
            2 \\
            3\\
            4
        \end{bmatrix},\]
        where we have $\sigma_1 = \sigma_2 = 2$ and $\operatorname{rank}(\mathbf{K}_{\mathbf{A,B}})=2.$
        By the proposed strategy, we take $N_0=2$ and the matrix
        \[\mathbf{B_1} = \begin{bmatrix}
             2 & 1 & 1\\
            3 & 1 & 1\\
            4 & 1 & 1
        \end{bmatrix} \qquad \Longrightarrow \qquad \operatorname{rank}(\mathbf{K}_{\mathbf{A,B_{1}}})=3,\]

        i,e., the pair $(\mathbf{A, B})$ is controllable. But for instance, we can also take
        \[\mathbf{B_2} = \begin{bmatrix}
             2 & 1 \\
            3 & 1\\
            4 & 1
        \end{bmatrix}\qquad \Longrightarrow \qquad \operatorname{rank}(\mathbf{K}_{\mathbf{A,B_{2}}})=3.\]
        Repeating the same idea with

        \[\mathbf{B_3} = \begin{bmatrix}
             2 & 1 \\
            3 & 0\\
            4 & 0
        \end{bmatrix}, \qquad \mathbf{B_4} = \begin{bmatrix}
             2 & 0 \\
            3 & 0\\
            4 & 1
        \end{bmatrix},\]
        we get  $\operatorname{rank}(\mathbf{K}_{\mathbf{A,B_{3}}})=3$ and $\operatorname{rank}(\mathbf{K}_{\mathbf{A,B_{4}}})=2$.
    \end{remark}
\end{enumerate}

To conclude this part, we study a very common situation in control theory for PDEs. Can we take a localization of the form, $b(x)=\mathds{1}_{\omega}$ where $\omega \subset (0,L)$ is a non-empty open set? We just mention that in the case $m=1$ and simple eigenvalues, this is not always possible with this approach. Following  \cite[Lemma 2.1]{Cerpa_2010}, we know that for $\sigma>0$ the eigenfunctions of $\mathcal{A}$ are of the form
\[e(x) = A\cos\left(\alpha\left( x-\frac{L}{2}\right)\right) + B\sin\left(\alpha\left(x-\frac{L}{2}\right)\right) + C \cos\left(\beta\left(x-\frac{L}{2}\right)\right) + D\sin\left(\beta\left(x-\frac{L}{2}\right)\right),\]
where $\alpha ,\beta \in \mathbb{R}$ are given by
\[\alpha = \sqrt{\frac{-\lambda - \sqrt{\lambda^2-4\sigma}}{2}}, \qquad \beta = \sqrt{\frac{-\lambda + \sqrt{\lambda^2-4\sigma}}{2}}.\]
By imposing the boundary conditions, we focus on the case
\[\beta \cos(\beta/2)\sin(\alpha/2)=\alpha\cos(\alpha/2)\sin(\beta/2), \quad \sin(\alpha/2)\not=0.\]
Thus, we get
\[e_j(x) = C\left[-\frac{\sin(\beta/2)}{\sin(\alpha/2)}\sin\left(\alpha\left(x-\frac{L}{2}\right)\right) + \sin\left(\beta\left(x-\frac{L}{2}\right)\right)\right].\]
The localization coefficients are then given by $b_j(x) = \langle \mathds{1}_\omega, e_j\rangle_{L^2(0,L)} = \int_\omega e_j(x)dx$. If we take $\omega = \left(\frac{L}{4}, \frac{3L}{4}\right)\subset (0,L)$, and observing the symmetry of $e_{j}$ with respect to $x=\frac{L}{2}$ we get
\[b_j(x) =\int_{L/4}^{3L/4} e_j(x)dx = 0, \qquad \det(\mathbf{K}_{\mathbf{A,B}})=0.\]
This shows that, even when the eigenvalues are simple, the choice of $\omega\subset (0,L)$ depends on the structure of the eigenfunctions.

\section{\texorpdfstring{$H^{2}$ and fully nonlinear Stability}{H2 and fully nonlinear stability}}
\label{sec: hg nl}
In the previous sections, we dealt with the exponential stability in the $L^{2}-$norm. To obtain a stability result for the fully nonlinear system (which includes the nonlinearity $\mathcal{N}(y)$), we now prove exponential stability in the $H^{2}-$norm. Motivated by \cite{guzman2019stabilization}, we consider the following Lyapunov functional:
\beq\label{eq: Lya V2}V_{2}(t)=\dfrac{M}{2}V_{1}(t)-\langle y,\mathcal{A}y\rangle_{L^{2}(0,L)}=\dfrac{M}{2}z^{\top}\mathbf{P}z-\langle y,\mathcal{A}y\rangle_{L^{2}(0,L)},\eeq
where $z=\iota(\Pi_{n}y)$ corresponds to the finite dimensional component of $y$ solution of \eqref{eq: main lineal} and $M>0$ big enough to fix later.
\begin{remark}
The term $-\langle y,\mathcal{A}y\rangle_{L^{2}(0,L)}$ may seem unusual in the structure of the Lyapunov functional $V_{2}$. This Lyapunov functional has a frequency-based; that is, it separates low and high frequency behaviour. For instance if, $\mathbf{P}=\mathbf{I}$ and $y=\dis\sum_{j=1}^{\infty}y_{j}e_{j}$, we get
\[V_{2}(t)=\dfrac{M}{2}\sum_{j=1}^{n}\left(\dfrac{M}{2}-\sigma_{j}\right)y_{j}^{2}+\sum_{j=1}^{\infty}\left(-\sigma_{j}\right)y_{j}^{2}.\]
\end{remark}

The goal of this section is to establish the following result:
\begin{theorem}
\label{th: sta H2}
Let $\mathbf{K}$, $\mathbf{P}$ be as in \Cref{prop: region fed}. For any $y_{0}\in H^{2}(0,L)$ such that $\iota(\Pi_{n}y_{0})\subset S=\mathcal{R}$, where $\mathcal{R}$ is the region of attraction of the finite dimensional system. The unique solution of the closed-loop system \eqref{eq: LKS closed loop} satisfies
\[\|y(t,\cdot)\|_{H^{2}(0,L)}\leq C e^{-at}\|y_{0}\|_{H^{2}(0,L)}, \qquad t>0.\]
\end{theorem}
\begin{proof}
We start by noting that, due to the boundary conditions \eqref{eq: bc KS}, it is enough to show the exponential decay of $\|\partial_{x}^{2}y(t,\cdot)\|_{L^{2}(0,L)}$. Indeed, we have
\beq\label{eq: equiv norm}\begin{aligned}
\|\partial_{x}y(t,\cdot)\|_{L^{2}(0,L)}^{2}&=\int_{0}^{L}\partial_{x}y(t,\cdot)\partial_{x}y(t,\cdot)dx\\
&=-\int_{0}^{L}y(t,\cdot)\partial_{x}^{2}y(t,\cdot)dx\\
&\leq\dfrac{1}{2}\left( \|y(t,\cdot)\|_{L^{2}(0,L)}^{2}+\|\partial_{x}^{2}y(t,\cdot)\|_{L^{2}(0,L)}^{2}\right).
\end{aligned}\eeq

\begin{lemma}
\label{lem: step 1}
There exist $C_{1},C_{2}>0$ such that
\begin{align}
\label{eq: step 1}
\dfrac{C_{1}}{2}|z|^{2}+\dfrac{C_{1}}{2C_{2}}\|\partial_{x}^{2}y(t,\cdot)\|_{L^{2}(0,L)}\leq V_{2}(t).
\end{align}
\end{lemma}
\begin{proof}
Let $\beta_{\min},\beta_{\max}>0$ be the minimum and maximum eigenvalues of the positive defined and symmetric matrix $\mathbf{P}$. Therefore, we have $V_{1}(t)\geq \beta_{\min}|z|^{2}$. On the other hand, since $\sigma_{1}\geq \sigma_{j}$ for all $j=1,\dots,n$ and that $z$ is solution of finite dimensional \eqref{eq: z sat}
\[\begin{aligned}
-\dfrac{1}{2}\sum_{j=1}^{\infty}\sigma_{j}y_{j}(t)^{2}&=-\dfrac{1}{2}\sum_{j=1}^{n}\sigma_{j}y_{j}(t)^{2}-\dfrac{1}{2}\sum_{j=n+1}^{\infty}\sigma_{j}y_{j}(t)^{2}\\
&\geq -\dfrac{\sigma_{1}}{2}|z|^{2}-\dfrac{1}{2}\sum_{j=n+1}^{\infty}\sigma_{j}y_{j}(t)^{2}.
\end{aligned}\]
Therefore,
\beq\label{eq: V_2 >}\begin{aligned}
V_{2}(t)&=-\dfrac{1}{2}\sum_{j=1}^{\infty}\sigma_{j}y_{j}(t)^{2}+\dfrac{M}{2}V_{1}(t)\\
&\geq \dfrac{(M\beta_{\min}-\sigma_{1})}{2}|z|^{2}-\dfrac{1}{2}\sum_{j=n+1}^{\infty}\sigma_{j}y_{j}(t)^{2}\\
&\geq C_{1}\left(|z|^{2}-\sum_{j=n+1}^{\infty}\sigma_{j}y_{j}(t)^{2}\right), \qquad C_{1}=\min\left\{\dfrac{(M\beta_{\min}-\sigma_1)}{2},\dfrac{1}{2}\right\}.
\end{aligned}\eeq
 Similarly to before, and recalling that $-e^{\prime\prime\prime\prime}_{j}-\lambda e^{\prime\prime}_{j}=\sigma_{j}e_{j}$, for all $j \in \mathbb{N}$
\[\begin{aligned}
\|\partial_{x}^{2}y(t,\cdot)\|_{L^{2}(0,L)}^{2}&=\sum_{(j,k)\in\mathbb{N}^{2}}y_{j}(t)y_{k}(t)\langle e_{j}^{\prime\prime},e_{k}^{\prime\prime}\rangle_{L^{2}(0,L)}\\
&=\sum_{(j,k)\in\mathbb{N}^{2}}y_{j}(t)y_{k}(t)\langle e_{j},e_{k}^{\prime\prime\prime\prime}\rangle_{L^{2}(0,L)}\\
&=\sum_{(j,k)\in\mathbb{N}^{2}}y_{j}(t)y_{k}(t)\langle e_{j},\lambda e^{\prime\prime}_{k}-\sigma_{k}e_{k}\rangle_{L^{2}(0,L)}\\
&=-\lambda\langle y(t,\cdot),\partial_{x}^{2}y(t,\cdot)\rangle_{L^{2}(0,L)}-\sum_{k=1}^{\infty}\sigma_{j}y^{2}_{j}(t)\\
&\leq\dfrac{1}{2}\left( \lambda^{2}\|y(t,\cdot)\|_{L^{2}(0,L)}^{2}+\|\partial_{x}^{2}y(t,\cdot)\|_{L^{2}(0,L)}^{2}\right)-\sum_{k=1}^{\infty}\sigma_{j}y^{2}_{j}(t).
\end{aligned}\]
Next, we find
\[\begin{aligned}
\|\partial_{x}^{2}y(t,\cdot)\|_{L^{2}(0,L)}^{2}&\leq\lambda^{2}\|y(t,\cdot)\|_{L^{2}(0,L)}^{2}-2\sum_{k=1}^{\infty}\sigma_{j}y^{2}_{j}(t)\\
&\leq \lambda^{2}\sum_{k=1}^{\infty}y^{2}_{j}(t)-2\sum_{k=n+1}^{\infty}\sigma_{j}y^{2}_{j}(t)\\
&=\lambda^{2}\sum_{k=1}^{n}y^{2}_{j}(t)+\sum_{k=n+1}^{\infty}(\lambda^{2}-2\sigma_{j})y^{2}_{j}(t)\\
&=\lambda^{2}|z|^{2}+\sum_{k=n+1}^{\infty}(\lambda^{2}-2\sigma_{j})y^{2}_{j}(t).
\end{aligned}\]

\noi Now, for $j\geq n+1$ we have $\sigma_j <0$ and $0< -\sigma_{n+1}\leq -\sigma_{j}$, thus

\begin{equation*}
    0<\lambda^2-2\sigma_j = -\sigma_j\Bigg(2-\frac{\lambda^2}{\sigma_j}\Bigg) \leq -\sigma_j \Bigg(2-\frac{\lambda^2}{\sigma_{n+1}}\Bigg).
\end{equation*}
Therefore, we get the following estimate
\begin{equation*}\label{H2 superior}
    \|\partial_{x}^{2}y(t,\cdot)\|_{L^2(0,L)}^2 \leq C_2\Bigg(|z|^2 - \sum_{j=n+1}^\infty \sigma_j y^{2}_j(t)\Bigg),
\end{equation*}

where $C_2 = \max\left\{\lambda^2, 2-\frac{\lambda^2}{\sigma_{n+1}}\right\}>0$. Similarly

\begin{equation}\label{estimacioninferior}
    -\sum_{j=n+1}^\infty \sigma_j y^{2}_j(t) \geq \frac{\|\partial_{x}^{2}y(t,\cdot)\|_{L^2(0,L)}^2}{C_2} -\left|z\right|^2.
\end{equation}

\noi Finally, from the previous computations we obtain

\[\begin{aligned}
V_2(t) &\geq C_1 |z|^2 -C_1 \sum_{j=n+1}^\infty \sigma_jy^{2}_j(t)\\
&\geq C_1 |z|^2 -\frac{C_1}{2} \sum_{j=n+1}^\infty \sigma_jy^{2}_j(t)\\
&\geq  C_1 |z|^2 +\frac{C_1\|\partial_{x}^{2}y\|_{L^2(0,L)}^2}{2C_2} - \frac{C_1\left|z\right|^2}{2}\\
&=\frac{C_1}{2}\left|z\right|^2 + \frac{C_1}{2C_2}\|\partial_{x}^{2}y\|_{L^2(0,L)},
\end{aligned}\]
from where we deduce  \eqref{eq: step 1}.
\end{proof}

\begin{lemma}
\label{lem: step 2}
There exists $a>0$ such that
$$V_2(t) \leq e^{-a t}\,V_2(0), \qquad \forall t\geq 0.$$
\end{lemma}

\begin{proof}
From \Cref{prop: region fed}, we know that the finite dimensional system \eqref{eq: z sat}
satisfies \eqref{eq: V1}. On the other hand, as $\mathcal{A}$ is self-adjoint we have
\[\dfrac{d}{dt}\langle\mathcal{A}y, y\rangle_{L^2(0,L)}=2\langle\mathcal{A}y, \partial_{t}y\rangle_{L^2(0,L)}.\]
Using \eqref{eq: y ope} with control $u=\mathbf{K}z$ we see

\begin{equation}
\label{eq: dt A}
\begin{aligned}
 -\frac{d}{dt}\langle\mathcal{A}y,y\rangle_{L^2(0,L)} &= -2\left\langle\mathcal{A}y, \mathcal{A}y + \sum_{k=1}^m b_{k}\cdot \sat{\left(\mathbf{K}z\right)_k}\right\rangle_{L^2(0,L)}\\
    &= -2\|\mathcal{A}y\|_{L^2(0,L)}^2 -2 \left\langle \mathcal{A}y, \sum_{k=1}^m b_{k}\cdot \sat{(\mathbf{K}z)_k}\right\rangle_{L^2(0,L)}.
\end{aligned}
\end{equation}

The second term in the above inequality can be estimated as follows
\[\begin{aligned}
    -2\left\langle \mathcal{A}y, \sum_{k=1}^m b_{k}\cdot \sat{(\mathbf{K}z)_k}\right\rangle_{L^2(0,L)} &\leq \|\mathcal{A}y\|^2_{L^2(0,L)} + \left|\mathbf{b}_j|^2\cdot |\sat{
    \mathbf{K}z}\right|^2\\
    &\leq \|\mathcal{A}y\|^2_{L^2(0,L)} + \|\mathbf{K}\|^2\sum_{k=1}^m \|b_k\|^2_{L^2(0,L)}|z|^2.
\end{aligned}\]
Using this information in \eqref{eq: dt A}

\begin{align}
\begin{split}
   -\frac{d}{dt}\langle\mathcal{A}y,y\rangle_{L^2(0,L)} &\leq -\|\mathcal{A}y\|^2_{L^2(0,L)}  + \|\mathbf{K}\|^2\sum_{k=1}^m \|b_k\|^2_{L^2(0,L)}|z|^2.\\
    &=-\sum_{j=1}^\infty \sigma_j^2y^{2}_j(t) + \|\mathbf{K}\|^2\sum_{k=1}^m ||b_k||^2_{L^2(0,L)}|z|^2.\label{-d/dt Aw w}
\end{split}
\end{align}

Observe that the eigenvalues $\{\sigma_j\}_{j\in\mathbb{N}}$ satisfy the following

\begin{equation*}
    \begin{cases}
        \sigma_1 > \sigma_j \geq 0, \quad & j\leq n\\
        0>\sigma_{n+1} \geq \sigma_j, \quad &j\geq n+1.
    \end{cases}
\end{equation*}

Moreover, if $j\geq n+1$ we have

\begin{equation}\label{Norma sigma cuadrado}
    -\sigma_j \leq \frac{\sigma_j^2}{-\sigma_{n+1}}, \qquad \forall j\geq n+1,
\end{equation}
and therefore, taking
\begin{equation}\label{Condicion M2}
    M\geq -\frac{1}{\sigma_{n+1}} ,
\end{equation}

\noi we see from \eqref{Condicion M2} and \eqref{Norma sigma cuadrado} that
\[\begin{aligned}
M\|\mathcal{A}y\|_{L^2(0,L)}^2 &=M \sum_{j=1}^\infty \sigma_j^2y^{2}_j(t)\\
&\geq M\sum_{j=n+1}^\infty \sigma_j^2y^{2}_j(t)\\
&\geq -\frac{1}{\sigma_{n+1}} \sum_{j=n+1}^\infty \sigma_j^2y^{2}_j(t)\\
&\geq -\sum_{j=n+1}^\infty \sigma_jy^{2}_j(t).
\end{aligned}\]
From the last line we deduce
\begin{equation}\label{-Aw^2}
    -\|\mathcal{A}y\|_{L^2(0,L)}^2 \leq \frac{1}{M} \sum_{j=n+1}^\infty \sigma_jy^{2}_j(t) \leq \frac{1}{M} \sum_{j=1}^\infty \sigma_jy^{2}_j(t).
\end{equation}

\noi Gathering \eqref{-Aw^2}, \eqref{-d/dt Aw w} and \eqref{eq: V1}

\begin{equation}
\label{Cotaa V2}
 \begin{aligned}
    \dot{V}_2(t) &= -\frac{1}{2}\frac{d}{dt}\langle\mathcal{A}y, y\rangle_{L^2(0,L)} + \frac{M}{2}\dot{V}_1(t)\\
    &\leq -\frac{1}{2} \sum_{j=1}^\infty \sigma^2_jy^{2}_j(t) + \frac{\left(\|\mathbf{K}\|^2\displaystyle\sum_{k=1}^m \|b_k\|^2_{L^2(0,L)}\right)}{2}\left|z\right|^2 - \frac{\alpha M}{2}\left| z\right|^2\\
    &=-\frac{1}{2}\|\mathcal{A}y\|^2_{L^2(0,L)} + \frac{\left(\|\mathbf{K}\|^2\displaystyle\sum_{k=1}^m \|b_k\|^2_{L^2(0,L)} -\alpha M\right)}{2}\left|z\right|^2\\
    &\leq \frac{1}{2M} \sum_{j=1}^\infty \sigma_jy^{2}_j(t)+  \frac{\left(\|\mathbf{K}\|^2\displaystyle\sum_{k=1}^m \|b_k\|^2_{L^2(0,L)}-\alpha M\right)}{2}|z|^2.
\end{aligned}
\end{equation}

As $\alpha>0$, we can take
\begin{equation}\label{C3}
C_3 > \frac{1}{2\alpha},
\end{equation}

and $M$ such that
\begin{equation}\label{Cota C3}
    \|\mathbf{K}\|^2\displaystyle\sum_{k=1}^m \|b_k\|^2_{L^2(0,L)}-\alpha M < -\frac{M}{2C_3},
\end{equation}
to obtain
\begin{equation}\label{Cota V2}
    \dot{V}_2(t) \leq \frac{1}{2M} \sum_{j=1}^\infty \sigma_jy^{2}_j(t) -\frac{M}{4C_3}\left| z\right|^2.
\end{equation}
\begin{remark}
From \eqref{C3} and estimate \eqref{Cota C3}, we deduce the inequality
\[0 <\frac{\|\mathbf{K}\|^2\displaystyle\sum_{k=1}^m \|b_k\|^2_{L^2(0,L)}}{\left(\alpha - \frac{1}{2C_3}\right)} < M,\]
that ensure that the value of $C_3>0$ is feasible and do not affect the sign of $M$.
\end{remark}

On the other hand, by definition of $\beta_{\text{max}}$
\begin{equation*}
-\left| z\right|^2 \leq -\frac{1}{\beta_{\text{max}}} {z^\top \mathbf{P}z},
\end{equation*}

using this in \eqref{Cota V2}, we get

 \begin{align*}
     \dot{V}_2(t) &\leq \frac{1}{2M} \sum_{j=1}^\infty \sigma_jy^{2}_j(t) -\frac{M}{4C_3}\left| z\right|^2\\
     &\leq \frac{1}{2M} \sum_{j=1}^\infty \sigma_jy^{2}_j(t) -\frac{1}{2C_3\beta_{\text{max}}}\frac{M}{2}\left| z\right|^2.
 \end{align*}

Finally, if $M\geq 2C_3\beta_{max}$, we obtain

 \beq\label{cota V2 2}\begin{aligned}
     \dot{V}_2(t) &\leq \frac{1}{2M} \sum_{j=1}^\infty \sigma_jy^{2}_j(t) -\frac{1}{2C_3\beta_{\text{max}}}\frac{M}{2}\left| z\right|^2\\
     &\leq \frac{1}{2C_3\beta_{\text{max}}}\left(\frac{1}{2}\sum_{j=1}^\infty \sigma_jy^{2}_j(t) - \frac{M}{2}z^\top \mathbf{P}z\right)\\
     &=-\frac{1}{2C_3\beta_{\text{max}}} \left(-\frac{1}{2}\langle\mathcal{A}y,y\rangle_{L^2(0,L)} + \frac{M}{2}V_1(t)\right)\\
     &=-\frac{1}{2C_3\beta_{\text{max}}} V_2(t).
 \end{aligned}\eeq

which finishes the proof of \Cref{lem: step 2} with
$a = \frac{1}{2C_3\beta_{\text{max}}}.$
\end{proof}
We conclude the proof of the exponential stability of $V_{2}$ by invoking the following lemma,
\begin{lemma}\label{lem: step 3}
There exists $C_4>0$ such that
    \[V_2(0) \leq C_4\|y_0\|^2_{H^2(0,L)}.\]
\end{lemma}
\begin{proof}
Using the expression of $\mathcal{A}$ and its boundary conditions \eqref{eq: bc KS} we observe that
\begin{align*}
    V_2(t) &= -\frac{1}{2}\langle\mathcal{A}y,y\rangle_{L^2(0,L)} + \frac{M}{2}V_1(t)\\
    &\leq \frac{1}{2}\|\partial_{x}^{2}y(t,\cdot)\|^2_{L^2(0,L)} -\frac{\lambda}{2}\|\partial_{x}y(t,\cdot)\|^2_{L^2(0,L)} + \frac{M\beta_{\text{max}}}{2}\left|z\right|^2\\
    &\leq \frac{1}{2}\|\partial_{x}^{2}y(t,\cdot)\|^2_{L^2(0,L)} + \frac{M\beta_{\text{max}}}{2}\|y(t,\cdot)\|^2_{L^2(0,L)}\\
    &\leq C_{4}\|y(t,\cdot)\|^2_{H^2(0,L)}.
\end{align*}
We conclude by taking $t=0$.
\end{proof}
Finally, \Cref{th: sta H2} comes immediately from \Cref{lem: step 1,lem: step 2,lem: step 3}.
\end{proof}
\subsection{Nonlinear case}
\label{sec: nonlinear}
In this section, we study the stabilization of the nonlinear system \eqref{eq: KS} under saturated feedback law. More specifically, we show the local exponential stability of the following system

\beq\label{eq: main closed-loop nl}
\begin{cases}
\partial_{t}y+\lambda \partial_{x}^{2}y+\partial_{x}^{4}y+\mathcal{N}(y)=\dis\sum_{k=1}^{m}b_{k}(x)\sat\left(\left(\mathbf{K}\iota^{-1}(\Pi_{n}y(\cdot,t))\right)\right)_{k},& t>0, \ x \in (0,L), \\
y(0,x)=y_{0}(x), & x \in (0,L).
\end{cases}\eeq
Similarly to the linear case  expanded every solution $y(t,\cdot) \in D(\mathcal{A})$ of \eqref{eq: main closed-loop nl} as a series on the eigenfunctions, $e_{j}(\cdot)$ we get that \eqref{eq: main closed-loop nl} is equivalent to the next infinity dimensional control system

\beq\label{eq: infinite node}
\dot{y}_{j}(t)=\sigma_{j}y_{j}(t)y+\mathbf{b}_{j}\sat\left(\left(\mathbf{K}\iota^{-1}(\Pi_{n}y(\cdot,t))\right)\right)+f_{j}(t), \quad j \in \mathbb{N},
\eeq

where $f_{j}(t)=\langle -\mathcal{N}(y),e_{j} \rangle_{L^{2}(0,L)}$. Recalling that  $n \in \mathbb{N}$ is the number of unstable eigenvalues of $\mathcal{A}$, we define

\beq\label{eq: F cont} \mathbf{F}(t)=\begin{pmatrix}
f_{1}(t)\\
\vdots \\
f_{n}(t)
\end{pmatrix},\eeq
we can construct the next finite dimensional system:
\beq\label{eq: node sat} \dot{z}(t)=\mathbf{A}z(t)+\mathbf{B}\sat(\mathbf{K}z(t))+\mathbf{F}(t). \eeq

Now we state our main result.

\begin{theorem}
\label{th: sta y nl}Consider the system \eqref{eq: main closed-loop nl} and let $\mathbf{K}$, $\mathbf{P}$ as in \Cref{prop: region fed}, then there exists $\eps>0$ such that for all $y_{0} \in D(\mathcal{A})$ satisfying $\|y_{0}\|_{H^{2}(0,L)}\leq \eps$, the unique solution $y \in C([0,T];H^{2}(0,L))\cap L^{2}(0,T,H^{4}(0,L))$ for all $T>0$ satisfy

\beq\label{eq: stb H2}
\|y(t,\cdot)\|_{H^{2}(0,L)}\leq Ce^{-at}\|y_{0}\|_{H^{2}(0,L)},\eeq

for $C>0$.
\end{theorem}

\begin{proof}
By \Cref{sec: well posedness}, we know that \eqref{eq: main closed-loop nl} correspond to \eqref{eq: LKS o KS closed loop} with $\delta=1$. Then, well-posedness part is a consequence of \Cref{prop: wp nlgKS} and \eqref{eq: dH2 full nl} and the stability. The proof of the stability part of \Cref{th: sta y nl} is based on the Lyapunov function \eqref{eq: Lya V2}. Let $T>0$, we calculate the derivative of $V_{2}$ along the trajectories of \eqref{eq: infinite node}

\[\dot{V}_{2}(t)=\dfrac{M}{2}\dfrac{d}{dt}\left(z^{\top}\mathbf{P}z\right)-\dfrac{d}{dt}\langle y,\mathcal{A}y\rangle_{L^{2}(0,L)}\]

We will analyze these two terms. First using \eqref{eq: node sat}

\[\dfrac{d}{dt}\left(z^{\top}\mathbf{P}z\right)=(\mathbf{A}z+\mathbf{B}\sat(\mathbf{K}z)+\mathbf{F})^{\top}\mathbf{P}z+z^{\top}\mathbf{P}(\mathbf{A}z+\mathbf{B}\sat(\mathbf{K}z)+\mathbf{F}).\]

By Proposition \eqref{prop: region fed} we know that in the set $\mathcal{R}$

\[(\mathbf{A}z+\mathbf{B}\sat(\mathbf{K}z))^{\top}\mathbf{P}z\\
+z^{\top}\mathbf{P}(\mathbf{A}z+\mathbf{B}\sat(\mathbf{K}z))\leq -2\alpha|z|^{2}.\]

Thus, we get
\beq\label{eq: z der nl}\dfrac{d}{dt}\left(z^{\top}\mathbf{P}z\right)\leq-\alpha|z|^{2}+2\mathbf{F}^{\top}\mathbf{P}z.\eeq

Then, using Young's inequality for $\eps_{1}>0$ to define later

\beq\label{eq: z nl est}\left|\mathbf{F}^{\top}\mathbf{P}z\right| \leq \|\mathbf{P}\||\mathbf{F}||z| \leq\dfrac{\|\mathbf{P}\|}{2}\left(\dfrac{1}{\eps_{1}}|\mathbf{F}|^{2}+\eps_{1}|z|^{2}\right).\eeq
Using the orthonormality of $e_{j}$, we observe
\[\begin{aligned}
|f_{j}(t)|&=\left|\int_{0}^{L}f(t,x)e_{j}(x)dx\right|\\
&\leq \|f(t,\cdot)\|_{L^{2}(0,L)}\\
&=\left\|y\partial_{x}y\right\|_{L^{2}(0,L)}.
\end{aligned}\]
By the continuous injection from $H^{1}(0,L)$ to $L^{\infty}(0,L)$ and \eqref{eq: equiv norm}
\[\begin{aligned}
\left\|y\partial_{x}y\right\|_{L^{2}(0,L)}&\leq\|y\|_{L^{\infty}(0,L)}\|\partial_{x}y\|_{L^{2}(0,L)}\\
&\leq \|y\|^{2}_{H^{1}(0,L)}\\
&= C\left(\|\partial_{x}y\|^{2}_{L^{2}(0,L)}+\|y\|^{2}_{L^{2}(0,L)}\right)\\
&\leq \dfrac{C}{2}\left(\|\partial_{x}^{2}y\|^{2}_{L^{2}(0,L)}+3\|y\|^{2}_{L^{2}(0,L)}\right)\\
&=\dfrac{C}{2}\left(\|\partial_{x}^{2}y\|^{2}_{L^{2}(0,L)}+3|z|^{2}+3\sum_{j=n+1}^{\infty}y^{2}_{j}(t)\right),
\end{aligned}\]
the term $\displaystyle\sum_{j=n+1}^{\infty}y^{2}_{j}(t)$ can be estimated as \eqref{-Aw^2} and using \eqref{eq: V_2 >} obtaining
\[\begin{aligned}
\sum_{j=n+1}^{\infty}y^{2}_{j}(t)&=\sum_{j=n+1}^{\infty}\dfrac{-\sigma_{j}}{-\sigma_{j}}y^{2}_{j}(t)\\
&\leq\sum_{j=n+1}^{\infty}\dfrac{-\sigma_{j}}{-\sigma_{n+1}}y^{2}_{j}(t)   \\
&\leq \dfrac{1}{-\sigma_{n+1}C_{1}}V_{2}-\dfrac{1}{-\sigma_{n+1}}|z|^{2}
\end{aligned}\]
Finally, joining these previous estimates with \eqref{eq: step 1} we get the existence of $C_{5}>0$ such that
\beq\label{eq: f_j est}
|f_{j}(t)|\leq C_{5}V_{2}(t).
\eeq
\begin{remark}
If we consider boundary conditions such that $e_{j} \in H_{0}^{1}(0,1)$, the term $f_{j}$ can be alternatively  estimated as follows
\[\begin{aligned}
f_{j}(t)&=\int_{0}^{L}y\partial_{x}ye_{j}dx\\
&=\dfrac{1}{2}\int_{0}^{L}(\partial_{x}|y|^{2})e_{j}dx\\
&=-\dfrac{1}{2}\int_{0}^{L}|y|^{2}e^{\prime}_{j}dx.
\end{aligned}\]
Thus, $|f_{j}(t)|\leq \|e_{j}\|_{W^{1,\infty}(0,L)}\|y(t,\cdot)\|_{L^{2}(0,L)}^{2}$.
\end{remark}
Using \eqref{eq: f_j est}, \eqref{eq: z der nl} and \eqref{eq: z nl est} we obtain

\beq\label{eq: Z der}\begin{aligned}\dfrac{d}{dt}\left(z^{\top}\mathbf{P}z\right)&\leq-\alpha|z|^{2}+\dfrac{\|\mathbf{P}\|C_{5}^{2}}{2\eps_{1}}V^{2}_{2}+\dfrac{\|\mathbf{P}\|\eps_{1}}{2}|z|^{2}\\
&\leq-\dfrac{\alpha}{2}|z|^{2}+\dfrac{\|\mathbf{P}\|^{2}C_{5}^{2}}{2\alpha}V^{2}_{2},\end{aligned}\eeq
where, we have taken $\eps_{1}=\frac{\alpha}{\|\mathbf{P}\|}$. On the other hand, similar to \eqref{eq: dt A}, we have that

\[\begin{aligned}
-\dfrac{1}{2}\dfrac{d}{dt}\langle y,\mathcal{A}y\rangle_{L^{2}(0,L)}=-\|\mathcal{A}y\|_{L^2(0,L)}^2 -\left\langle \mathcal{A}y, \sum_{k=1}^m b_{k}\cdot \sat{(\mathbf{K}z)_k}\right\rangle_{L^2(0,L)}-\langle f,\mathcal{A}y\rangle_{L^{2}(0,L)},
\end{aligned}\]

Using the same ideas as for the first term and with \eqref{Cota V2}, we get from Cauchy inequality

\beq\label{eq: A der}\begin{aligned}
-\dfrac{1}{2}\dfrac{d}{dt}\langle y,\mathcal{A}y\rangle_{L^{2}(0,L)}&\leq-\dfrac{1}{2}\|\mathcal{A}y\|_{L^2(0,L)}^2+\frac{\left(\|\mathbf{K}\|^2\displaystyle\sum_{k=1}^m \|b_k\|^2_{L^2(0,L)}\right)}{2}\left|z\right|^2+\dfrac{C_{5}^{2}}{2}V^{2}_{2}\\
&\leq \dfrac{1}{4M}\sum_{j=1}^{\infty}\sigma_{j}y^{2}_{j}(t)+\frac{\left(\|\mathbf{K}\|^2\displaystyle\sum_{k=1}^m \|b_k\|^2_{L^2(0,L)}\right)}{2}\left|z\right|^2+\dfrac{C_{5}^{2}}{2}V^{2}_{2}.
\end{aligned}\eeq
Combining \eqref{eq: Z der} and \eqref{eq: A der} we derive

\[\begin{aligned}\dot{V_{2}}(t) \leq \dfrac{1}{4M}\sum_{j=1}^{\infty}\sigma_{j}y^{2}_{j}(t)+\frac{\left(2\|\mathbf{K}\|^2\displaystyle\sum_{k=1}^m \|b_k\|^2_{L^2(0,L)}-M\alpha\right)}{4}\left|z\right|^2+\dfrac{C_{5}^{2}}{2\alpha}\left(\|\mathbf{P}\|^{2}+\alpha\right)V^{2}_{2}.\end{aligned}\]

Then, if we take $M$ big enough, and following the \eqref{Cota C3}, \eqref{Cota V2} and \eqref{cota V2 2}, we can conclude the existence of $A$, $B>0$ such that

\[\dot{V_{2}}(t) \leq -AV_{2}(t)+BV^{2}_{2}(t), \text{ for all } t\in [0,T).\]

We can argue that if $V_{2}(0)\leq A/(2B)$, with $V_{2}(0)\neq0$, then $V_{2}(t)\leq(A/B)e^{-At}$ for all $t\in[0,\infty)$. To that end, we  apply \Cref{ineq} with $V_{2}(t)=v(t)$, $b(t)=-A$, $k(t)=B$ and $p=2$. Then, in virtue of \eqref{ineq_2} we get
\[w(t)\geq V_{2}(0)^{-1}-(B/A)\left(1-e^{-At}\right)\geq V(0)^{-1}-B/A.\]
Accordingly, if we choose $V_{2}(0)>0$ so that $V_{2}(0)^{-1}-B/A\geq B/A$, which happens when $V_{2}(0)\leq A/(2B)$, it follows that $w(t)\geq B/A>0$ for all $t\in[0,\infty)$. Therefore, from \eqref{ineq_3} we deduce the required inequality:

\beq\label{eq: V decay}
V_{2}(t)\leq(A/B)e^{-At}, \quad t\in[0,T).\eeq

To finish the proof of the stabilization, note that by \Cref{lem: step 3}

\[V_{2}(0)\leq C_{4}\|y_{0}\|^{2}_{H^{2}(0,L)}.\]

Finally taking $\eps>0$ small enough such that if $\|y_{0}\|^{2}_{H^{2}(0,L)}\leq \eps$ then $z(0) \in \mathcal{R}$ and $V_{2}(0)\leq A/(2B)$ we conclude \eqref{eq: stb H2}. This stability estimate gives us if $\|y_{0}\|^{2}_{H^{2}(0,L)}\leq \varepsilon$, then $\|y(t,\cdot)\|^{2}_{L^{2}(0,L)}\leq C\|y_{0}\|^{2}_{H^{2}(0,L)}$. Thus, from \eqref{eq: dH2 full nl} we get the global in time well-posedness in $C([0,T];H^{2}(0,L))\cap L^{2}(0,T;H^{4}(0,L))$.



\end{proof}

\section{Boundary feedback stabilization}
\label{sec: bs}
In this section, we briefly outline how the sames ideas extend to boundary feedback. The analysis is similar to the internal case, with only minor modifications concerning the boundary operator. Again, the main result of this section holds for the three class of boundary conditions. For simplicity, we write the proof in the case \textit{clamped} case and KS nonlinearity. We consider the following system
\begin{equation}
\label{eq: KS bdr}
\begin{cases}
\partial_{t}y+\lambda \partial_{x}^{2}y+\partial_{x}^{4}y+y\partial_{x}y=0,& t>0, \ x \in (0,L), \\
y(t,0)=y(t,L)=\partial_{x}y(t,L)=0,    &t>0, \\
\partial_{x}y(t,0)=u(t), &  t>0,\\ 
y(0,x)=y_{0}(x), & x \in (0,L), \\
\dot{u}(t)=\sat(h(t)),    &t>0.
\end{cases}
\end{equation}
We also consider the linearized system subject to boundary saturation
\begin{equation}
\label{eq: LKS bdr}
\begin{cases}
\partial_{t}y+\lambda \partial_{x}^{2}y+\partial_{x}^{4}y=0,& t>0, \ x \in (0,L), \\
y(t,0)=y(t,L)=\partial_{x}y(t,L)=0,    &t>0, \\
\partial_{x}y(t,0)=u(t),  & t>0,\\ 
y(0,x)=y_{0}(x), & x \in (0,L), \\
\dot{u}(t)=\sat(h(t)),    &t>0.
\end{cases}
\end{equation}

Note that in the systems \eqref{eq: KS bdr}-\eqref{eq: LKS bdr} the input control is $h(t)$ but acts on the system by $\dot{u}(t)=\sat(h(t))$ that means that our control acts as an integrator. The main result of this section, namely \Cref{th: main bdr}, states the local exponential stability of \eqref{eq: KS bdr} with saturated boundary control. Many of the computation of this section are quite similar to those presented in \Cref{sec: caso memoria} and for that reason are omitted.
\medskip

Inspired by \cite{Cerpa_2010,guzman2019stabilization} we define $d(x)=\left(\frac{x^{3}}{L^{2}}-\frac{2x^{2}}{L}+x\right)$ and we make the change of variable $w(t,x)=y(t,x)-d(x)u(t)$. From now, we suppose that $u \in H^{1}(0,T)$ for all $T>0$. Easy calculations show that $w$ satisfies the following system

\beq\label{eq: Lwsys}
\begin{cases}
w_{t}=\mathcal{A}w+a(x)u(t)+b(x)\sat(h(t)),& t>0, \ x \in (0,L), \\
w(t,0)=w(t,L)=0,    &t>0, \\
\partial_{x}w(t,0)=\partial_{x}w(t,L)=0,    &t>0, \\
w(0,x)=y_{0}-d(x)u(0), & x \in (0,L), \\
\dot{u}(t)=\sat(h(t)),    &t>0,\\
\end{cases}\eeq
where $a(x)=-\lambda d^{''}(x)$ and $b(x)=-d(x)$.

\begin{remark}
Due to the parabolic behaviour, there is no a big difference if the control acts at left or right boundary, this differs for instance the case of dispersive equations \cite{cerpa2013rapid}. In the case of different boundary conditions, we can consider as lifting function. In general, the lifting function is built as a third-order polynomial that vanish at the not acted boundary conditions, and it is equal to $1$ where the control acts.
\end{remark}

\begin{remark}
We take $\dot{u}(t)=\sat(h(t))$ and not directly $u(t)=\sat(h(t))$. This is motivated because $\sat$ operator is not enough smooth for making the change of variable $w=y-du$.
\end{remark}

Similar to \Cref{sec: caso memoria}, expanding $w(t,\cdot) \in D(\mathcal{A})$ of \eqref{eq: Lwsys} as a series in the eigenfunctions basis $\{e_{j}\}_{j\in \mathbb{N}}$, we observe that \eqref{eq: Lwsys} is equivalent to the next infinity dimensional control system

\beq\label{eq: infinite ode bdr} \begin{cases}
 \dot{u}(t)=\sat(h(t)), \\
\dot{w_{j}}(t)=\sigma_{j}w_{j}(t)+a_{j}u(t)+b_{j}\sat(h(t)), \quad j \in \mathbb{N},
\end{cases}\eeq

where $a_{j}=\langle a,e_{j} \rangle_{L^{2}(0,L)}$ and $b_{j}=\langle b, e_{j} \rangle_{L^{2}(0,L)}$. Let $n \in \mathbb{N}$ the number of unstable eigenvalues of $\mathcal{A}$ and  $\eta>0$ as in \eqref{eq: neg eig}. Defining the matrices:
\beq\begin{gathered}  z(t)=\begin{pmatrix} u(t) \\
w_{1}(t)\\
\vdots \\
w_{n}(t)
\end{pmatrix}, \ \ \mathbf{A}=\begin{pmatrix} 0 & \cdots & \cdots & 0 \\
a_{1} & \sigma_{1} & \cdots & 0 \\
\vdots & \vdots & \ddots & 0 \\
a_{n} & 0 & \cdots & \sigma_{n}\\
\end{pmatrix}, \ \ \mathbf{B}=\begin{pmatrix} 1 \\
b_{1}\\
\vdots \\
b_{n}
\end{pmatrix}. \end{gathered}
\eeq
We can construct the next unstable finite dimensional system:
\beq\label{eq: ode sat bdr} \dot{z}(t)=\mathbf{A}z(t)+\mathbf{B}\sat(h(t)). \eeq
In this case, we have the following proposition about the controllability of the pair $(\mathbf{A},\mathbf{B})$
\begin{proposition}(Proposition 1, \cite{Cerpa_2010})
\label{prop: kalman}
If $\lambda\notin \mathcal{N}$, then the pair $(\mathbf{A},\mathbf{B})$ is stabilizable.
\end{proposition}

Along this section, we work under the assumption
\begin{assumption}
The anti-difussion parameter $\lambda\notin \mathcal{N}$, where, $\mathcal{N}$ is defined in \eqref{eq: critical}.
\end{assumption}

As the pair $(\mathbf{A},\mathbf{B})$ is stabilizable,  there exists $\mathbf{K}=\begin{pmatrix}
k_{0}&k_{1}& \cdots & k_{n}
\end{pmatrix} \in \mathbb{R}^{1\times(n+1)}$, such that $\mathbf{A}+\mathbf{BK}$ is Hurwitz. \Cref{prop: kalman} motivates the choice of the control $h(t)=\mathbf{K}z(t)$.

\beq\label{eq: z sat bdr} \dot{z}(t)=\mathbf{A}z(t)+\mathbf{B}\sat(\mathbf{K}z(t)), \eeq
As in the internal case, we can prove the local exponential stability
\begin{proposition}
There exists $\mathbf{P}\in\mathbb{R}^{n+1\times n+1}$ a symmetric positive definite, such that \eqref{eq: z sat bdr} is locally exponentially stable in $0$ with region of attraction $\mathcal{R}$ given by
\beq\label{eq: R bdr} \mathcal{R}:= \left\{z\in \mathbb{R}^{n+1} \ : \ z^{\top}\mathbf{P}z \leq 1 \right\}. \eeq

Moreover, in $\mathcal{R}$, the function $V_{1}(z):=z^{\top}\mathbf{P}z$, $z \in \mathbb{R}^{n+1}$, decreases exponentially fast to $0$ along the solutions to \eqref{eq: z sat bdr}, i.e., there is a constant $\alpha>0$ so that

\beq\label{eq: V1 bdr} \dot{V_{1}}(z)\leq -\alpha|z|^{2}, \quad z \in \mathcal{R}.\eeq
\end{proposition}
\begin{remark}
Note that, the control expression is not too explicit at the first glance. In fact $\dot{u}(t)=\sat\left(\mathbf{K}z(t)\right)=\sat\left(k_{0}u(t)+\dis\sum_{j=1}^{n}k_{j}w_{j}(t) \right)$. But the system \eqref{eq: z sat bdr} is well-defined because $\sat$ is Lipschitz continuous.
\end{remark}

In this direction, we can prove the following result with respect to the exponential stability of \eqref{eq: infinite ode bdr}

\begin{proposition}
\label{prop: sta Lw bdr}
Let $\mathbf{K}$ as in \Cref{prop: kalman} and $S=S_{0}\times S_{n} \subset \mathbb{R}^{n+1}$ the region of attraction of \eqref{eq: z sat bdr} such that the system is locally exponentially stable in $0$, where $S_{0}\subset\mathbb{R}$ and $S_{n}\subset\mathbb{R}^{n}$. Then:
\begin{enumerate}
\item system \eqref{eq: infinite ode bdr} with the feedback law $h(t)=\mathbf{K}z(t)$ is locally exponentially stable in $0$ with region of attraction $S\times \ell_{2, j>n}$.

\item system \eqref{eq: Lwsys} with the feedback law $h(t)=\mathbf{K}z(t)$ is locally exponentially stable in $0$ with region of attraction $S_{0}\times\iota(S_{n})\times X_{n}^{\perp}$.
\end{enumerate}

In addition, for any closed and bounded set $G \subset \inte(S_{0}\times\iota(S_{n})\times X_{n}^{\perp})$, there exists two positive values $M$, $\alpha$ such that for any initial condition $w(0,\cdot) \in G$, the solution $(u,w)$ of \eqref{eq: Lwsys} with controller $h(t)=\mathbf{K}z(t)$ satisfies

\beq\label{eq: decay L2 bdr} |u(t)|+\|w(t,\cdot)\|_{L^{2}(0,L)} \leq Ce^{-at}\|w(0,\cdot)\|_{L^{2}(0,L)}.\eeq
\end{proposition}
\begin{proof}
The proof is straightforward, following \Cref{sec: caso memoria}.
\end{proof}
\begin{remark}
We search for region of attraction of the form $S=S_{0}\times S_{n}$ for simply the notation. In \cite{Mironchenko_2020} and, the past section works with internal controller, which implies that the state is only given by the projections $w_{1},\dots,w_{n}$.
\end{remark}

\begin{remark}
Our control is given by $\dot{u}(t)=\mathbf{K}z(t)$, thus we need to prescribe the initial condition $u(0)$. Without loss of generality,  we can take  $u(0)=0$, which implies $w(0,\cdot)=y_{0}$.
\end{remark}
By the previous remark, from \eqref{eq: decay L2 bdr} we obtain the $L^{2}$ decay of $y=w+ud$
\[\begin{aligned}
\|y(t,\cdot)\|_{L^{2}(0,L)}&\leq \|w(t,\cdot)\|_{L^{2}(0,L)}+|u(t)|\|d\|_{L^{2}(0,L)}\\
&\leq \max\left\{1,\|d\|_{L^{2}(0,L)}\right\}\left(w(t,\cdot)\|_{L^{2}(0,L)}+|u(t)|\right)\\
&\leq C\max\left\{1,\|d\|_{L^{2}(0,L)}\right\}e^{-at}\|w(0,\cdot)\|_{L^{2}(0,L)}\\
&=C\max\left\{1,\|d\|_{L^{2}(0,L)}\right\}e^{-at}\|y_{0}\|_{L^{2}(0,L)}.
\end{aligned}\]

Following \Cref{th: sta H2}, the $H^{2}$ stability can also be shown in the boundary case. This allows us to prove the main result of this section.
\begin{theorem}
\label{th: main bdr}
Let $\mathbf{K}$ as in \Cref{prop: kalman}, $\mathbf{P}$ as in \eqref{eq: R bdr} and $h(t)=\mathbf{K}z$, then there exists $\eps>0$ such that for all $y_{0} \in H_{0}^{2}(0,L)$ satisfying $\|y_{0}\|_{H^{2}_{0}(0,L)}\leq \eps$, the unique solution $y \in C([0,T];H^{2}(0,L))\cap L^{2}(0,T,H^{4}(0,L))$ for all $T>0$ satisfy

\beq\label{eq: stb H2 bdr}
\|y(t,\cdot)\|_{H^{2}_{0}(0,L)}\leq Ce^{-\alpha t}\|y_{0}\|_{H^{2}_{0}(0,L)},\eeq

for $C>0$.
\end{theorem}

\section{Conclusions and perspectives}
\label{sec: conclusion}
The present work establishes a framework for the stabilization of fourth-order parabolic equations with input saturation for internal and boundary feedback laws. Our approach relied on spectral decomposition and modal stabilization using LMIs. We established local exponential stability in the $L^{2}$ and $H^2$ norms. Using this high regularity decay, we prove that the fully nonlinear systems can be stabilized with saturated actuator.  A natural perspective for future research is to extend these results to boundary actuation without relying on spectral lifting arguments. This could be achieved by exploiting the theory of admissible control operators for analytic semigroups, which provides a direct way to handle unbounded boundary inputs. Such an approach would allow the development of a unified methodology applicable to  saturated stabilization for general parabolic operators, including the multidimensional case.

\appendix
\section{Estimates local in time well posedness}
\label{sec: app ltwp}
Let $w\in B_{\theta}$, then
\beq\label{eq: F est}
\begin{aligned}
\|\mathcal{F}(w)\|^{2}_{L^{2}(0,\theta;L^{2}(0,L))}&=\int_{0}^{\theta}\int_{0}^{L}\left(\sum_{k=1}^{m}b_{k}(x)\sat\left(\left(\mathbf{K}\iota^{-1}(\Pi_{n}w(\cdot,t))\right)\right)_{k}\right)^{2}dxdt \\
&\leq\sum_{k=1}^{m}\|b_{k}\|^{2}_{L^{2}(0,L)}\int_{0}^{\theta}\left(\sat\left(\left(\mathbf{K}\iota^{-1}(\Pi_{n}w(\cdot,t))\right)\right)_{k}\right)^{2}dt\\
&\leq \left(\sum_{k=1}^{m}\|b_{k}\|^{2}_{L^{2}(0,L)}\right)\|\mathbf{K}\|^{2}\int_{0}^{\theta}|\iota^{-1}(\Pi_{n}w(\cdot,t))|^{2}dt\\
&\leq C\theta\|w\|_{L^{2}(0,\theta;L^{2}(0,L))}^{2},\\
&\leq C\theta\|w\|^{2}_{B_{\theta}}.
\end{aligned}\eeq
For the nonlinear term $G$, we have
\[\|\mathcal{G}(w)\|_{L^{2}(0,\theta;L^{2}(0,L))}\leq \lambda\|\partial_{x}^{2}w\|_{L^{2}(0,\theta;L^{2}(0,L))}+\delta\|w\partial_{x}w\|_{L^{2}(0,\theta;L^{2}(0,L))}+\delta\|\partial_{x}^{2}(w^{3})\|_{L^{2}(0,\theta;L^{2}(0,L))},\]
we study these three terms separately
\begin{itemize}
    \item For the first one,
    \[\begin{aligned}
     \|\partial_{x}^{2}w\|_{L^{2}(0,\theta;L^{2}(0,L))}^{2}&=\int_{0}^{\theta}\int_{0}^{L}|\partial_{x}^{2}w|^{2}dxdt\\
    &\leq\theta\|w\|_{C^{2}([0,\theta];H^{2}(0,L))}^{2}\\
    &\leq\theta\|w\|_{B_{\theta}}^{2}.
    \end{aligned}\]

    \item For the second one, using the injection of $H^{1}(0,L)$ into $L^{\infty}(0,L)$
    \[\begin{aligned}
     \|w\partial_{x}w\|_{L^{2}(0,\theta;L^{2}(0,L))}^{2}&=\int_{0}^{\theta}\int_{0}^{L}|w\partial_{x}w|^{2}dxdt\\
    &\leq \int_{0}^{\theta}\|\partial_{x}w(t,\cdot)\|_{L^{\infty}(0,L)}^{2}\|w(t,\cdot)\|_{L^{\infty}(0,L)}^{2}dt\\
    &\leq \int_{0}^{\theta}\|w(t,\cdot)\|_{H^{2}(0,L)}^{4}dt\\
    &\leq C\theta\|w\|_{C^{2}([0,\theta];H^{2}(0,L))}^{4}\\
    &\leq C\theta\|w\|_{B_{\theta}}^{4}.
    \end{aligned}\]
    \item For the last one,
    \[\begin{aligned}
     \|\partial_{x}^{2}(w^{3})\|_{L^{2}(0,\theta;L^{2}(0,L))}^{2}=&\int_{0}^{\theta}\int_{0}^{L}|6w(\partial_{x}w)^{2}+3w^{2}\partial_{x}^{2}w|^{2}dxdt\\
    \leq& C\left(\int_{0}^{\theta}\int_{0}^{L}w^{2}(\partial_{x}w)^{4}dxdt+\int_{0}^{\theta}\int_{0}^{L}w^{4}(\partial_{x}^{2}w)^{2}dxdt\right)\\
    \leq& C\left(\int_{0}^{\theta}\|w(t,\cdot)\|_{L^{\infty}(0,L)}^{2}\|\partial_{x}w(t,\cdot)\|_{L^{\infty}(0,L)}^{4}dt\right.\\
    &\left.+\int_{0}^{\theta}\|w(t,\cdot)\|_{L^{\infty}(0,L)}^{4}\|\partial_{x}^{2}w(t,\cdot)\|_{L^{2}(0,L)}^{2}dt\right)\\
    \leq& C\int_{0}^{\theta}\|w(t,\cdot)\|_{H^{2}(0,L)}^{6}dt\\
    \leq& C\theta\|w\|_{C^{2}([0,\theta];H^{2}(0,L))}^{6}\\
    \leq& C\theta\|w\|_{B_{\theta}}^{6}.
    \end{aligned}\]
\end{itemize}

Joining all the estimates, we get
\beq\label{eq: G est}
\|\mathcal{G}(w)\|_{L^{2}(0,\theta;L^{2}(0,L))}\leq C\sqrt{\theta}\left(\|w\|_{B_{\theta}}+\delta\|w\|_{B_{\theta}}^{2}+\nu\|w\|_{B_{\theta}}^{3}\right).\eeq

\bibliographystyle{plain}
\bibliography{ref}

\end{document}